\documentclass[letterpaper,11pt]{article}
\usepackage[margin=1in]{geometry}

\usepackage[utf8]{inputenc}
\usepackage{times}
\usepackage{booktabs} 
\usepackage{tabulary} 
\usepackage[flushleft]{threeparttable} 
\usepackage{amsthm, amsfonts, amsmath}
\usepackage{mathtools}
\usepackage{algorithm}
\usepackage{algorithmicx}
\usepackage[noend]{algpseudocode}
\usepackage{enumerate}
\usepackage{xcolor}
\usepackage{xspace}
\usepackage{subcaption}

\usepackage[textsize=tiny]{todonotes}

\definecolor{darkgreen}{rgb}{0,0.5,0}
\definecolor{darkblue}{rgb}{0,0,0.8}
\usepackage{hyperref}

\hypersetup{
	colorlinks=true,                  
	linkcolor=darkblue,
	citecolor=darkgreen
}
\usepackage[capitalize,nameinlink]{cleveref} 

\newcommand{\hide}[1]{}

\newcommand{\phil}[1]{\todo[color=red!50]{Phil: #1}}

\newcommand{\LOCAL}{\ensuremath{\mathsf{LOCAL}}\xspace}
\newcommand{\CONGEST}{\ensuremath{\mathsf{CONGEST}}\xspace}

\newcommand{\SPD}{\ensuremath{\mathsf{SPD}}\xspace}

\newcommand{\id}{\ensuremath{\text{id}}}

\newcommand{\bigO}{\smash{\ensuremath{O}}}
\newcommand{\tilO}{\smash{\ensuremath{\widetilde{O}}}}
\newcommand{\tilOm}{\smash{\ensuremath{\widetilde{\Omega}}}}
\newcommand{\tilT}{\smash{\ensuremath{\widetilde{\Theta}}}}

\newcommand{\eps}{\varepsilon}
\newcommand{\calE}{\mathcal{E}}
\newcommand{\calB}{\mathcal{B}}
\renewcommand{\Pr}{\mathbb{P}}
\newcommand{\E}{\mathbb{E}}

\newcommand{\set}[1]{\ensuremath{\left\{#1\right\}}}

\DeclareMathOperator{\polylog}{polylog}
\DeclareMathOperator{\poly}{poly}

\DeclareMathOperator*{\argmin}{arg\,min}
\newboolean{short}
\setboolean{short}{false}

\newcommand{\shortOnly}[1]{\ifthenelse{\boolean{short}}{#1}{}}
\newcommand{\onlyShort}[1]{\ifthenelse{\boolean{short}}{#1}{}}
\newcommand{\longOnly}[1]{\ifthenelse{\boolean{short}}{}{#1}}
\newcommand{\onlyLong}[1]{\ifthenelse{\boolean{short}}{}{#1}}

\begin{document}
\title{Shortest Paths in a Hybrid Network Model}

\author{
John Augustine\\
IIT Madras, India\\
{\small\texttt{augustine@iitm.ac.in}}
\and
Kristian Hinnenthal\\
Paderborn University, Germany\\
{\small\texttt{krijan@mail.upb.de}}
\and
Fabian Kuhn\\
U.\ of Freiburg, Germany\\
{\small\texttt{kuhn@cs.uni-freiburg.de}}
\and
Christian Scheideler\\
Paderborn University, Germany\\
{\small\texttt{scheideler@upb.de}}
\and
Philipp Schneider\\
U.\ of Freiburg, Germany\\
{\small\texttt{philipp.schneider@cs.uni-freiburg.de}}
}

\date{}
\maketitle

\vspace{1cm}







\newtheorem{theorem}{Theorem}[section]
\newtheorem{proposition}[theorem]{Proposition}
\newtheorem{lemma}[theorem]{Lemma}
\newtheorem{corollary}[theorem]{Corollary}
\newtheorem{observation}[theorem]{Observation}

\newtheorem{remark}[theorem]{Remark}
\newtheorem{fact}[theorem]{Fact}
\newtheorem{definition}[theorem]{Definition}

\vspace{-15mm}

\begin{abstract}\normalsize
    We introduce a communication model for \emph{hybrid networks}, where nodes have access to two different communication modes: a \emph{local} mode where (like in traditional networks) communication is only possible between specific pairs of nodes, and a \emph{global} mode where (like in overlay networks) communication between any pair of nodes is possible. 
    This can be motivated, for instance, by wireless networks in which we combine direct device-to-device communication (e.g., using WiFi) with communication via a shared infrastructure (like base stations, the cellular network, or satellites).

    Typically, communication over short-range connections is cheaper and can be done at a much higher rate. Hence, we are focusing here on the $\mathsf{LOCAL}$ model (in which the nodes can exchange an unbounded amount of information in each round) for the local connections while for the global communication we assume the so-called \emph{node-capacitated clique model}, where in each round every node can exchange only $O(\log n)$-bit messages with just $O(\log n)$ other nodes.
    However, our model for hybrid networks is very general so that it can also capture many other scenarios, like the congested clique model.
    
    In order to explore the power of combining local and global communication, we study the complexity of computing shortest paths in the graph given by the local connections. 
    We show that our model allows the development of algorithms that are significantly faster than what can be done by using either local or global communication only. 

    We specifically show the following results. 
    For the all-pairs shortest paths problem (APSP), we show that an exact solution can be computed in time $\tilde O\big(n^{2/3}\big)$\footnote{Note that the $\tilde O(\cdot)$-notation hides factors that are polylogarithmic in $n$.} and that approximate solutions can be computed in time $\tilde \Theta\big(\!\sqrt{n}\big)$. 
    For the single-source shortest paths problem (SSSP), we show that an exact solution can be computed in time \smash{$\tilde O\big(\!\sqrt{\mathsf{SPD}}\big)$}, where $\mathsf{SPD}$ denotes the shortest path diameter. 
    We further show that a \smash{$\big(1\!+\!o(1)\big)$}-approximate solution can be computed in time $\tilde O\big(n^{1/3}\big)$.
    Additionally, we show that for every constant $\varepsilon>0$, it is possible to compute an $O(1)$-approximate solution in time $\tilde O(n^\varepsilon)$.


\end{abstract}

\pagenumbering{gobble}

\newpage

\pagenumbering{arabic}

 \section{Introduction}
\label{sec:intro}

Many existing communication networks exploit a combination of multiple communication modes to maximize cost-efficiency and throughput.
As a prominent example, hybrid datacenter networks combine high-speed optical or wireless circuit switching technologies with traditional electronic packet switches to offer higher throughput at lower cost  \cite{FPRBSFPV10,HKPBW11}. 
In the Internet, dynamic multipoint VPNs can be set up to connect different branches of an organization by combining leased lines (offering them quality-of-service guarantees for their mission-critical traffic) with standard, best-effort VPN connections (for their lower-priority traffic) \cite{RS11}. 
Alternatively, an organization may also set up a so-called hybrid WAN by combining their own communication infrastructure with connections via the Internet \cite{TBKC18}. 
Finally, the emerging 5G standard promises to allow smartphones to not only communicate via the cellular infrastructure, but also directly with other smartphones via their wireless interface. 
This allows them to set up a hybrid network consisting of connections via base stations as well as device-to-device (D2D) connections \cite{KS18}, which is particularly interesting for vehicular networks.

Despite the advantages that have been experienced with hybrid communication networks in practice, rigorous theoretical research on hybrid networks is still in its infancy. In this paper, we propose a simple model for hybrid networks and explore its power for a fundamental problem in graph theory as well as communication networks: computing shortest paths.
More specifically, we present exact and approximate algorithms for the all-pair shortest paths problem and single-source shortest paths problem.

In our hybrid communication model, we assume that each node has two different communication modes: a {\em local} communication mode that allows it to send messages along each of its edges in the given (private, leased, trusted, or ad-hoc) communication network, and a {\em global} communication mode, which allows a node to send messages to any node in the network, but to only exchange a limited number of messages in each round using this mode.
We demonstrate that by making use of both local and global communication, we can achieve significant runtime improvements compared to using local or global edges alone.
In particular, we investigate ways to improve the running time of computing shortest paths in the local network with the help of global communication, which highlights the importance of exploiting the hybrid communication capabilities of modern networks. 
Before we present our detailed results in 
\Cref{sec:contributions}, we formally introduce our model. We conclude this section with a discussion of related work in \Cref{sec:related}.

\subsection{Hybrid Communication Model}
\label{sec:model}

We assume that we are given a fixed set $V$ of $n$ nodes that are connected via two kinds of edges: {\em local} edges and {\em global} edges\footnote{Throughout the paper, we assume that $n\geq n_0$ for a sufficiently large constant $n_0$.}.
The local edges form a fixed, undirected, weighted graph $G = (V,E,w)$, where the edge weights are given by $w:E \to \{1,\ldots,W\} \subseteq \mathbb{N}$ for some $W$ that is at most polynomial in $n$. Thus, every weight and
length of any shortest path can be represented using $O(\log n)$ bits. The graph $G$ is said to be {\em unweighted} if $w:E \to \{1\}$. The global edges form a clique, i.e., every node can potentially send a message to any other node with the help of a global edge. We assume that each node $u$ has a unique identifier $\id(u)$. For simplicity, we assume that the node identifiers are $1,\dots,n$.\footnote{Assuming that the IDs are from $1$ to $n$ is a strong assumption. However, all our results can be obtained in the same way if we assume that nodes have arbitrary $O(\log n)$-bit IDs and that there is a sampling service that allows the nodes to contact a node that is chosen approximately uniformly at random via a global edge.} 

We use the standard synchronous message passing model, where time is divided into synchronous {\em rounds}, and in each round every node can send messages of size $O(\log n)$ to other nodes using its local and global edges.
In the most general form of our model, the number of messages that can be sent over either type of edge is restricted by parameters $\lambda$ and $\gamma$:
the \emph{local capacity} $\lambda$ is the maximum number of messages that can be sent over each local edge in a round, and the \emph{global capacity} $\gamma$ is the maximum number of messages any node can send and receive via global edges in a round.
When in some round more than $\lambda$ (or $\gamma$) messages are sent over an edge (or to a node, respectively), we assume that an adversary delivers an arbitrary subset of these messages and drops the other messages.
All of our algorithms ensure that with high probability\footnote{An event holds with high probability (w.h.p.) if it holds with probability at least $1\!-\!\frac{1}{n^c}$ for an arbitrary but fixed constant $c>0$.}, a node never sends or receives too many messages.

Note that whereas $\lambda$ imposes a bound on the number of messages that can be sent over each \emph{edge}, $\gamma$ effectively restricts the amount of global communication at each \emph{node}.
This modeling choice is motivated by the idea that local communication rather relates to \emph{physical} networks, where an edge corresponds to a physical connection (e.g., cable- or ad-hoc networks), whereas global communication primarily captures aspects of \emph{logical} networks that are formed as an overlay on top of some shared physical infrastructure.
For appropriate choices of $\lambda$ and $\gamma$, our model captures various established network models:
\LOCAL ($\lambda = \infty, \gamma = 0$), \CONGEST ($\lambda = O(1), \gamma = 0$), \textit{Congested Clique}\footnote{The congested clique model refers to the unicast variant, i.e.\ the \CONGEST model on a clique topology.} ($\lambda = O(1), \gamma = 0$ and $G$ is a clique), and the recently introduced \emph{Node-Capacitated Clique} model ($\lambda = 0, \gamma = O(\log n)$) \cite{AGG+18}.

In order to demonstrate the power of combining local and global communication, we will focus the variant of the model in which local edges are fully uncapacitated ($\lambda = \infty$), and global communication is heavily restricted ($\gamma = O(\log n)$).
Thus, our model is a combination of the most permissive \LOCAL model for local edges and the very restrictive node-capacitated clique model for the global edges, which makes it particularly clean and well-suited to investigate the power of hybrid networks from a theoretical perspective.
Moreover, we believe that the practical relevance of this model is justified by the fact that direct connections between devices are typically highly efficient and offer a large bandwidth at comparatively low cost, whereas communication over a shared global communication network such as the Internet, satellites, or the cellular network, is costly and typically offers only a comparatively small data rate.

We remark that any choice of $\gamma$ in the range from $\Theta(1)$ to $\Theta(\log^c n)$ would not change our upper bounds since it would only affect them by polylogarithmic factors, which we mostly neglect by using the $\tilO(\cdot)$-notation. The maximally permissive choice of $\lambda = \infty$ is mainly for proving lower bounds (which we do for APSP) and our algorithms make no overly excessive use of it. In fact, we show that our algorithms work for some $\lambda$ between $\Theta(1)$ and $\Theta(n^2)$, depending on the algorithm. (We discuss the congestion of local edges at the end of each section.)

\subsection{Contributions}
\label{sec:contributions}

The overarching goal of this paper is to achieve significantly faster solutions for shortest path problems than what would be possible if only either the local or the global network could be used. Note that by just using the local network (\LOCAL model), all graph problems can trivially be solved in time at most $D$, where $D$ is the diameter of the graph $G$. Shortest path problems also clearly have an $\Omega(D)$ lower bound using only the local network (note that $D = \Omega(n)$ in the worst case). Our objective is to understand to what extent a limited amount of global communication (given by the global network) helps in solving shortest path problems faster. To that end, we will briefly discuss our contributions. A summary is given in \Cref{tab:ourContribution}.

\begin{table}[h]
	\centering
	\begin{threeparttable}
	\renewcommand{\arraystretch}{1.3}
	\begin{tabulary}{\textwidth}{@{\hskip 0.1mm}L@{\hskip 0.4cm}L@{\hskip 0.3cm}L@{\hskip 0.2cm}L@{\hskip 1cm}L@{\hskip 0.2cm}L@{\hskip 0.4cm}L@{\hskip 0.2cm}L@{}}
		\toprule
           \textbf{APSP} &&&& \textbf{SSSP} &&& \\
           Approx.\ & Weights & Complexity & Local Cap.$^{\!\dagger}$ & Approx.\ & Weights & Complexity & Local Cap.$^{\!\dagger}$ \\
       \midrule
           Exact & weighted & $\tilO(n^{2/3})$ & $\Theta(n)$ & Exact & weighted & $\tilO\big(\!\sqrt{\SPD}\big)$ & $\tilT\big(n^2\!/\!\sqrt{\SPD}\big)$\\
           $(1\!+\!\varepsilon)$ & unw. & $\tilO\big(\!\sqrt{n/\varepsilon}\big)$ & $\Theta(n)$ & $\big(1\!+\!\varepsilon\big)$ & weighted & $\tilO(n^{1/3}\!/\varepsilon^6)$ & $\tilT(n^{2/3}\varepsilon^6)$\\
           $3$ & weighted & $\tilO(\!\sqrt{n})$ & $\Theta(n)$ & $(1/\varepsilon)^{O(1/\varepsilon)}$ & weighted & $\tilO(n^{\varepsilon})$ & $\Theta(1)$\\
           $\Omega(\!\sqrt{n})$ & unw. & $\tilOm(\!\sqrt{n})$ & $\infty$ & $s(n)$$^\ddagger$ & weighted & $s(n)$$^\ddagger$ & $\Theta(1)$ \\
		\bottomrule
	\end{tabulary}
    \begin{tablenotes}
		\footnotesize
		\item $\dagger$ local capacity $\lambda$ for which the corresponding bound holds
		\item $\ddagger$ where $s(n) = 2^{O(\!\sqrt{\log n \log\log n})} = o(n^c)$ for arbitrary $c>0$ is a sub-polynomial function
	\end{tablenotes}
	\end{threeparttable}
	\caption{Overview of the contributions of this paper.}
	\label{tab:ourContribution}
\end{table}

\vspace*{-3mm}

\paragraph{Token Dissemination.}

First we consider the \textit{token dissemination} problem. It represents the task of broadcasting a set of tokens of size $O(\log n)$ bits, each of which is initially only known by one node. We develop a protocol, tailored to the hybrid model, which we use as a subroutine throughout the paper. The main idea behind the algorithm is to randomly disseminate the tokens via global edges. This is sufficient for each node to afterwards collect all the tokens in a relatively small neighborhood using local edges. The runtime compares favourably to the respective lower bounds of the problem in case only the global or local network would be available. Specifically, we show an upper bound of \smash{$\tilO\big(\!\sqrt{k} + \ell \big)$} where $k$ is the number of distinct tokens and $\ell$ is the initial maximum number of tokens per node (\Cref{thm:tokenDissemination}). Note that the lower bound is $\Omega\big((k\!+\!\ell)/\log n\big)$ if only the global network can be used and $\Omega(n)$ if only the local network is available.

\vspace*{-3mm}

\paragraph{All-Pairs Shortest Paths (APSP).} Our primitives to solve APSP are based on combining the token dissemination scheme with the classic approach~\cite{ullman91} of building skeleton graphs. The basic idea is to sample a set of nodes with some probability $\frac{1}{x}$ and then computing virtual edges among pairs of sampled nodes connected by a path of at most $h \in \tilO(x)$ hops. We then employ our token dissemination protocol to broadcast the distance information of the skeleton graph and also from all other nodes to the skeleton. This global knowledge can then be used to compute the distances of any pair of nodes with sufficient hop distance.

While the approach of using skeleton graphs is not new, we demonstrate how the amount of work on the local and global network can be balanced (with parameter $x$), leading to interesting exact and approximative results for APSP. Specifically, in \Cref{thm:APSP}, we show that APSP can be solved exactly with running time \smash{$\tilO(n^{2/3})$}. Furthermore, we obtain $3$-approximate distances in time \smash{$\tilO\big(\!\sqrt{n}\big)$} for general graphs (\Cref{thm:approxgeneralAPSP}) and $(1\!+\!\eps)$-approximate distances in time \smash{$\tilO\big(\!\sqrt{n/\eps}\big)$} for unweighted graphs (\Cref{thm:approxunweightedAPSP}). Note that this is significantly better than the $\tilde{\Omega}(n)$ bound if only either the local or the global network could be used. (These bounds immediately follow from the facts that the diameter of the local network might be $\Omega(n)$ and that every node can only receive $O(\log n)$ messages over global edges.)

Finally, we complement our upper bounds for APSP in the hybrid model with a lower bound by proving that even for computing an $\alpha$-approximate solution for some $\alpha = \tilO\big(\!\sqrt{n}\big)$, at least $\tilOm\big(\!\sqrt{n}\big)$ rounds are required (\Cref{thm:lowerBoundAPSPapprox}), showing that our approximate APSP algorithms are tight up to logarithmic factors.

\vspace*{-3mm}

\paragraph{Single-Source Shortest Paths (SSSP).} For the SSSP problem, we present three algorithms, based on different techniques. In \Cref{thm:exactSSSP}, we prove that the SSSP problem can be solved exactly in time \smash{$\tilO\big(\!\sqrt{\SPD}\big)$}, where \SPD denotes the shortest path diameter. We introduce a new technique, that first uses the local network so each node can learn the graph up to a distance of \smash{$2\sqrt{\SPD}$} hops. Then every node knows all trees up to depth \smash{$\sqrt{\SPD}$} rooted at any neighbor within \smash{$\sqrt{\SPD}$} hops. This knowledge can be used to distribute distance information to the source $s$ in an iterative fashion over the global network. In iteration $i$, all nodes that have a shortest path to $s$ with $O(i^2)$ hops learn their distance to $s$. An iteration takes only $\tilO(1)$ rounds, leveraging a divide and conquer approach on the aforementioned trees and the aggregation protocol of \cite{AGG+18}. Note that on unweighted graphs using only either local or global edges, the best known algorithms require $\Omega(\SPD)$ rounds (which is tight for local edges as $D = \SPD$ on unweighted graphs).

%

We then shift our attention to approximate solutions of the SSSP problem. We give a simple algorithm that simulates the broadcast congested clique model\footnote{In the broadcast congested clique model, the nodes of an $n$-node graph are connected to a clique and in each round, each node can broadcast an $O(\log n)$-bit message to all other nodes.} on a set of sampled skeleton nodes (including the source) using token dissemination. We employ the SSSP algorithm by Becker et al.~\cite{BKKL17} for said model as a black box to solve SSSP on the skeleton, which allows us to compute an $\big(1\!+\!o(1)\big)$ approximation of SSSP in time \smash{$\tilO(n^{1/3})$} (\Cref{thm:approximateSSSP}).

The third, technically most challenging SSSP algorithm is based on recursively building a hierarchy of $O(\log_\alpha n)$ (for some $\alpha>1$) \textit{skeleton spanners} (i.e., spanners of skeleton graphs). Roughly speaking, given some skeleton spanner $H$, we obtain the next coarser skeleton spanner $H'$ by sampling each node of $H$ with probability $1/\alpha$ and computing a spanner with a good stretch on the sampled nodes. As a technical result, we show that given a low arboricity graph $H$, we can efficiently compute a low arboricity spanner $H'$ of $H$ using only global edges. We show in \Cref{theorem:polySSSPSimple}, that by choosing $\alpha=n^{\eps}$ for some $\eps>0$, we can compute $(1/\eps)^{O(1/\eps)}$-approximate paths to the source node in time $\tilde{O}(n^{\eps})$. For any constant $\eps>0$, we get a constant SSSP approximation (albeit with a potentially large constant). Choosing $\varepsilon$ to balance time and approximation factor, the algorithm computes a (subpolynomial) $2^{O(\!\sqrt{\log n \log\log n})}$-approximate SSSP solution in the same time.

\vspace*{-3mm}

\paragraph{Congestion of Local Edges.} Even though our algorithms are optimized for run-time, we analyze the local capacity $\lambda$ for which the given upper bounds hold at the end of each section (\Cref{tab:ourContribution} gives a summary). We briefly discuss the bottlenecks of our algorithms. In its final step, our token dissemination protocol (\Cref{sec:tokenDissemination}) ensures that all tokens appear in a small neighborhood of any node, where we collect them via local edges. Each edge has to transmit at most $k$ tokens in total, which we can (figuratively speaking) spread evenly over the runtime $O(\!\sqrt{k})$ of the algorithm (we use a slight adaption of a standard scheduling technique, c.f.\ \cite{Gha15} Theorem 1.1). Thus the protocol works for local capacity $\lambda = \Theta(\!\sqrt{k})$. 

For our APSP algorithms (\Cref{sec:apsp}) nodes need to learn distances of shortest paths with at most $h$ hops, as we can use the skeleton only for shortest paths with more than $h$ hops. This can be done with $h$ rounds of the distributed Bellman-Ford, for which $O(n)$ distance labels are exchanged over each edge in each round, thus requiring $\lambda = \Theta(n)$. We also use skeleton graphs for our approximate SSSP algorithm (\Cref{sec:sssp_approx}). However, here nodes only require the lengths of shortest paths (with at most $h$ hops) to skeleton nodes $M$ and the source $s$, thus Bellman-Ford exchanges at most $|M| = \tilT(n^{2/3}\varepsilon^6)$ distance labels.

The most demanding in terms of local capacity is our exact SSSP algorithm (\Cref{sec:sssp_exact}). Since each node requires knowledge about \textit{all} trees in a certain neighborhood, each node has to learn the whole neighborhood. This neighborhood can be of size \smash{$O(n^2)$} thus we require \smash{$\lambda = \tilT\big(n^2\!/\!\sqrt{\SPD}\big)$} to transmit it. Contrary to that, our approximate SSSP result based on recursion (\Cref{sec:subpolynomial}) is the most undemanding in terms of local capacity. We use the local network only to compute the initial skeleton spanner, employing \cite{BS07} as a black box. Since this algorithm is conceptualized for the \CONGEST model, it requires only $\lambda = \Theta(1)$. The recursive computation of subsequent skeleton spanners relies exclusively on the global network.

\subsection{Related Work}
\label{sec:related}

In the systems area, research on hybrid networks has mostly focused on wireless mesh networks (see, e.g., \cite{AAB17} for a recent survey), with a plethora of competing schemes for routing packets in such a network, though many of them do not exploit the hybrid communication capabilities of such networks. Hybrid communication has also been studied in the context of data centers demonstrating that it can significantly improve their cost-effectiveness and performance (e.g., \cite{FPRBSFPV10,HKPBW11}).

In the theory area, only very few results are known so far on hybrid networks. Jung et al.~\cite{JKSS18} studied the problem of finding near-shortest routing paths in ad-hoc networks satisfying certain properties by combining communication via ad-hoc connections with communication via the cellular infrastructure. Furthermore, Foerster et al.~\cite{FGS18} investigated the computational complexity of exploiting a hybrid infrastructure in data centers. Beyond these two publications,  various results have recently been developed for network models that allow global communication and are thus remotely related to our model.

Recently, some alternative models that are closer to our model have been considered. For example, the work of Gmyr et al.~\cite{GHSS17} implies that making use of global edges with the same constraints as our model significantly improves the ability to monitor properties of the network formed by the local edges. Furthermore, Augustine et al.~\cite{AGG+18} propose the so-called node-capacitated clique model, which is identical to our model for the global edges (but which has no local edges). They present distributed algorithms for various fundamental graph problems, including problems such as computing an MST, a BFS tree, a maximal independent set, a maximal matching, or a vertex coloring of the given input graph. Their BFS tree construction can be used to solve the SSSP problem for unweighted graphs of bounded arboricity in $\tilO(D)$ time. 
As we will see in \Cref{sec:subpolynomial}, the same can also be achieved similarly for the weighted SSSP problem.

In the congested clique model, every node has an edge to every other node in the system and in each round, each node can exchange a distinct $O(\log n)$-bit message with each other node. Over the last few years, research on congested clique algorithms has been very active (see, e.g., \cite{DKO14,JN18,lenzen13,LSPP05} for a small subset of the work). 
In the context of shortest path problems, Lenzen et al.~\cite{BKKL17} presented a $(1+\varepsilon)$-approximation algorithm for the SSSP problem that runs in time $\polylog n$ and Nanongkai~\cite{Nan14} presented a $(2+o(1))$-approximation algorithm for the APSP problem with runtime $\tilO(\sqrt{n})$. The algorithms of \cite{BKKL17} and \cite{Nan14} even work for the broadcast variant of the congested clique where in each round, every node has to send the same $O(\log n)$-bit message to all other nodes.
Shortest path problems can also be approached by performing matrix multiplications efficiently~\cite{CKK15}:
Then, APSP for example, can be solved exactly in time $\tilde{O}(n^{1/3})$, and a $(1+o(1))$-approximation can be found in time $O(n^{0.158})$. A recent result shows that a $(2+o(1))$-approximation can even be computed in time $\polylog n$~\cite{CensorHillel19}.


Note that our hybrid network model contains the congested clique model as a special case, though for the specific case considered in our paper the results in the congested clique model are not of any help since, in general, it is very costly to emulate an algorithm for the congested clique model in our case.


Further, shortest path problems have been intensively studied in standard distributed communication models, most importantly in the \CONGEST model. Many of our algorithms for the hybrid network model employ ideas that have been developed in this context.
For the SSSP problem, Das Sarma et al.~\cite{DHKK12} showed that any distributed approximation algorithm has a runtime of $\tilOm(\sqrt{n}+D)$ for any constant approximation ratio. 
Following the publication of this lower bound, there has been a series of papers that attempt to obtain algorithms that get close to the lower bound (see, e.g., \cite{LP13,Nan14,HKN16}), culminating in the work of Becker et al.~\cite{BKKL17}, which gives an algorithm that computes a $(1+\varepsilon)$-approximate SSSP solution in time $\tilO(\sqrt{n}+D)$.
For the exact SSSP problem, no better upper bound than $O(n)$ was known until two years ago when Elkin~\cite{Elkin17} presented an algorithm with a runtime of $\tilO(n^{2/3} D^{1/3} + n^{5/6})$. This was further improved by Ghaffari and Li~\cite{GL18} and by Forster and Nanongkai~\cite{FN18}, who presented two protocols for polynomially bounded edge weights, one with runtime $\tilO(\sqrt{n \cdot D})$ and one with runtime $\tilO(\sqrt{n} D^{1/4} + n^{3/5} + D)$.

For the APSP problem, a deterministic $(1+o(1))$-approximation algorithm
with runtime $\tilO(n)$ is known, as well as a nearly matching lower bound of $\tilOm(n)$ that even holds for randomized $O(\poly(n))$-approximation algorithms, even when $D = O(1)$ \cite{LP13,LP15,Nan14}. The complexity of the exact unweighted version was shown to be $\tilT(n)$ \cite{LeP13,HW12,FHW12,PRT12,ACK16}. 
For the exact weighted version, the first improvement over the naive $O(m)$-time algorithm was due to Huang et al.~\cite{HNS17}, who presented a randomized $\tilO(n^{5/4})$-time algorithm that bears similarity to our approaches based on skeleton nodes (which they call ``centers''). Subsequently, Bernstein and Nanongkai \cite{BN18} came up with a randomized $\tilO(n)$-time algorithm, so also this case is now settled up to $\polylog(n)$ factors. The best deterministic algorithm for the weighted APSP problem is due to Agarwal et al.~\cite{ARKP18} and has a runtime of $\tilde{O}(n^{3/2})$ using a technique based on ``blocker sets''quite similar to ours based on skeleton nodes.

 \section{Overview}

In this section we would like to provide the reader with an intuitive explanation of our core concepts without the in-depth technical details of the subsequent sections. Additionally, we will state our main theorems and sketch some proof ideas. Let us start by introducing some basic definitions.

\subsection{Preliminaries and Problem Definitions}

The \emph{distance} between any two nodes $u,v \in V$ of a graph $G=(V,E)$ is defined as 
\vspace*{-1mm}
\[
d_G(u,v) := \!\!\min_{\text{$u$-$v$-path } P}\, w(P),\vspace*{-1mm}
\]
where $w(P) = \sum_{e \in P}w(e)$ denotes the length of a path $P \subseteq E$.
A path between two nodes with smallest length is called a \emph{shortest path}.
The \emph{hop-distance} between two nodes $u$ and $v$ is defined as
\vspace*{-1mm}
\[
hop_G(u,v) := \min_{\text{$u$-$v$-path } P} |P|,\vspace*{-1mm}
\]
where $|P|$ denotes the number of edges (or \emph{hops}) of a path $P$.
Let the \emph{$h$-limited distance} from $u$ to $v$:
\vspace*{-1mm}
$$d_{h,G}(u,v) := \!\!\min_{\substack{\text{$u$-$v$-path } P\\ |P| \leq h }}\, w(P).\vspace*{-1mm}$$
If there is no $u$-$v$ path $P$ with $|P|\leq h$, then $d_{h,G}(u,v) := \infty$. The \emph{diameter} $D(G)$ of $G$ is defined as the length of any shortest path in $G$, and the \emph{shortest-path diameter} $\SPD(G)$ is the minimum number such that $d_{\SPD(G),G}(u,v)=d_G(u,v)$ for all $u,v\in V$. 
Whenever the graph $G$ is clear from the context, we drop the subscript $G$ in the above notations. 
In this paper, we consider the following shortest-paths problems in $G$.

\vspace*{-4mm}

\paragraph{All-Pairs Shortest Paths Problem (APSP).} Every node $u \in V$ has to learn $d(u,v)$ for all $v \in V$. In the $\alpha$-approximate APSP problem for some $\alpha>1$, every node $u\in V$ has to learn values $d'(u,v)$ such that $d(u,v)\leq d'(u,v)\leq \alpha\cdot d(u,v)$ for all $v\in V$.

\vspace*{-4mm}

\paragraph{Single-Source Shortest Paths Problem (SSSP).} There is a source $s \in V$ and every node $u \in V$ has to learn $d(u,s)$. In the $\alpha$-approximate SSSP problem for some $\alpha>1$, every node $u\in V$ has to learn $d'(u,s)$ such that $d(u,s)\leq d'(u,s)\leq \alpha\cdot d(u,s)$.

\smallskip

In order to solve shortest path problems efficiently, we also show how to solve the \textit{$(k,\ell)$-token dissemination problem ($(k,\ell)$-TD)}. Here we are given a set of $k$ tokens each of size $O(\log n)$-bits that need to be learned by all nodes $v\in V$. Initially, each token is known by one node and no node initially possesses more than $\ell$ tokens.
As a byproduct of our exact SSSP algorithm, we solve the \textit{$h$-limited $k$-source shortest paths problem ($(h,k)$-SSP)}, in which there is a set $S \subseteq V$ of $k$ sources and a parameter $h \ge 1$ and every node $u \in V$ has to learn $d_h(u,s)$ for every $s \in S$. 

\subsection{Token Dissemination}

The first tool that we are introducing solves the token dissemination problem $(k, \ell)$-TD. 
The algorithm consists of four steps. First, we balance the number of tokens per node. Each node redistributes its tokens randomly via global edges for $\bigO(\ell)$ rounds such that afterwards, w.h.p., each node has at most $\tilO(k/n)$ tokens to take care of. This first step eliminates the dependency on $\ell$ in the subsequent steps.

Second, if $k \!\ll\! n$, we create $\tilO(n/k)$ copies of each token to make sure that each node has some token, which allows to speed up the subsequent third step. We increase the number of copies of each token in the network in an exponential fashion in $\log(n/k)$ phases. In each phase, every node sends two copies of the tokens it received in the previous phase to random nodes via global edges. Note that this works because the total number of copies remains in $\tilO(n)$ and therefore the contention on the global network is not too high.

Third, each node sends the tokens it knows so far via global edges to a random subset of $V$, so that afterwards each node possesses a given token with probability at least $1/x$. This takes only $\bigO(k/x)$ rounds, relying on the fact that for $k \ll n$ we have already done part of the work in the previous step. Afterwards, any subset of $V$ of size $\tilOm(x)$ contains all tokens, w.h.p.

Fourth, the local edges are used to learn the tokens of $\tilOm(x)$ nodes in the neighborhood (w.r.t.\ local edges) of any given node, which takes $\tilO(x)$ rounds. The parameter $x$ signifies the trade-off between the running time of the third and fourth step and is optimized accordingly (\smash{$x \!\in\! \tilO\big(\!\sqrt{k}\big)$}). In  \Cref{sec:tokenDissemination} we give the details of the algorithm (\Cref{alg:tokenDissemination}) and we provide a full proof of \Cref{thm:tokenDissemination}.


\begin{theorem}
	\label{thm:tokenDissemination}
	There is an algorithm that solves $(k, \ell)$-TD on connected graphs in $\tilO\big(\!\sqrt{k}  \!+\! \ell\big)$ rounds, w.h.p.
\end{theorem}

\subsection{Upper Bounds for Exact APSP}

The first step to solve APSP is to construct an overlay graph $S = (M,E_S)$ on $G$ that we call a \textit{skeleton} \cite{ullman91} and whose nodes $M \subseteq V$ are obtained by marking nodes of $V$ uniformly at random (with probability $1/x$, for some optimization parameter $x$). Two nodes in $M$ have an edge if their hop distance is at most $h$. The weight of such an edge is the $h$-limited distance between its endpoints.
To compute $E_S$, we explore a (small) $h$-hop-neighborhood around every node in the local network. Since we choose $h \in \tilT(x)$, this takes $\tilO(x)$ rounds. Afterwards, each node knows the $h$-limited distance between it and all other nodes.

After the local exploration, we do not have to worry about pairs of nodes for which a shortest path of at most $h$ hops exists, since for those pairs the $h$-limited distances equal the true distance. For pairs $u,v \in V$, for which all shortest $u$-$v$-paths have \textit{more} than $h$ hops, we show that on one such a path there is a skeleton node within every $h$ hops w.h.p.\ (\Cref{lem:longPathsM}). This is particularly helpful since these pairs can now compute their distance if they have knowledge of the distance information of the (sparse) set of skeleton nodes.

The expected size of $M$ is $|M| \in {\tilO(n/x)}$, and we show that the weights of all edges $E_S$ can be broadcast to the whole network in ${\tilO(n/x)}$ with the methods of \Cref{sec:tokenDissemination} (token dissemination). Equipped with that information about $E_S$, each node can locally compute the distance Matrix $D^S$ of the skeleton $S$, and we show that distances among nodes in $S$ equal those in $G$. 
Finally, we disseminate the $h$-limited distances between pairs in $M \times V \!\setminus\! M$, which can be done in $\tilO\big(n/\!\sqrt{x}\big)$. With this information, all nodes know the distance matrix $D'$ containing said $h$-limited distances between pairs $M \times V \!\setminus\! M$. Then any node in $V$ can locally compute the complete distance matrix $D^G$ of $G$ as follows (we set $D'_{uv} = 0$ if $u,v \in M$):

\begin{equation}
\label{eq:exactAPSP}
D^G_{uv} = \min \Big(d_h(u,v), \min_{u'\!,v' \in M} \big( D'_{uu'} \!+ D^S_{u'v'} \!+ D'_{vv'}\big)\Big).
\end{equation}
For $x \in [1..n]$ the running time for the exploration via local edges is $\tilO(x)$ and $\tilO\big(n/\!\sqrt{x}\big)$ for the dissemination of the distance matrices $D^S$ and $D'$. This is optimized for \smash{$x \in \Theta\big( n^{2/3}\big)$}. In \Cref{sec:apsp_exact} we present \Cref{alg:exactAPSP} and its subroutines and give the full proof of \Cref{thm:APSP}.

\begin{theorem}
	\label{thm:APSP}
	There is an algorithm that solves APSP in \smash{$\tilO(n^{2/3})$} rounds w.h.p.
\end{theorem}

\subsection{Upper Bounds for  Approximate APSP}

The bottleneck of the exact algorithm is the dissemination of $h$-hop limited distances between all pairs of nodes in \smash{$V \!\setminus\! M \times M$}, which takes \smash{$\tilO\big(n/\!\sqrt{x}\big)$} rounds. Our approximative approach (\Cref{alg:approxAPSP}) mitigates this bottleneck by disseminating only \textit{one} distance token $d_{uv}$ per $v \in V\setminus M$, representing the distance between $v$ and its \textit{closest} marked node $u \in M$, which can be done \smash{$\tilO\big(\!\sqrt{n}\big)$} rounds.
We show that for suitable choices of $x$, \Cref{alg:approxAPSP} can be used to obtain a 3-approximation of APSP for general (connected) graphs in \smash{$\tilO\big(\!\sqrt{n}\big)$} and a $(1\!+\!\varepsilon)$-approximation for unweighted graphs in \smash{$\tilO\big(\!\sqrt{n/\varepsilon}\,\big)$}. 

Besides slightly adapted procedures, \Cref{alg:approxAPSP} uses the same subroutines as in the exact case. After all subroutines are performed, each node in $V$ knows (w.h.p.): (i) its $m$-hop distances $d_m(\cdot,v)$ to any other node $v \in V$ (c.f. Fact \ref{fct:constructskeletonprime}) (ii) the skeleton $S$ and the distance matrix $D^S$ among its nodes (c.f.\ \Cref{lem:distancesSkeleton}) and (iii) the distance $d_{vv'}$ between any node $v \in V \!\setminus\! M$ and its respective closest marked node $v' \in M$ (c.f. Fact \ref{fct:transmitclosest}). With this knowledge, each node computes an approximative distance Matrix $\tilde{D}^G$ of $G$ as follows: 
\begin{equation}
\label{eq:approxAPSP}
\tilde{D}^G_{uv} = \min\Big(d_m(u,v), \min_{u'\in M} \big(d_h(u,u') + D^S_{u'v'}\big) + d_{vv'}\Big),
\end{equation}
where $v'\!\in\!M$ is $v$'s closest marked node and $m \!=\! \max\!\big(h,\frac{n}{h}\big)$. In \Cref{sec:apsp_approx} we prove the following:

\begin{theorem}
	\label{thm:approxgeneralAPSP}
	There is an algorithm to compute a 3-approximation of APSP in $\tilO\big(\!\sqrt{n}\big)$ rounds w.h.p.
\end{theorem}

\begin{theorem}\label{thm:approxunweightedAPSP}
	For arbitrary $\varepsilon > 0$, there is an algorithm that computes a $(1\!+\!\varepsilon)$-approximation of the APSP problem on unweighted graphs in $\tilO\big(\!\sqrt{n/\varepsilon} \,\big)$ rounds w.h.p.
\end{theorem}

\subsection{Lower Bounds for APSP}

In order to obtain rigorous lower bounds we introduce the technical \Cref{lem:lowerboundlemma} in \Cref{sec:apsp_lower_bounds}. It shows that for a class of graphs and a dedicated node $b$, we can create a bottleneck for the information that can be transmitted from parts of the graph to $b$. More specifically, we show that if the state of some random variable $X$ is given to the nodes of some subgraph $G'$ and if $b$ is at the end of some path of length $L$, then every randomized algorithm in which $b$ needs to learn the state of $X$ requires \smash{$\tilOm\big(\!\min(L,{H(X)}/{L})\big)$} rounds, where $H(X)$ denotes the Shannon entropy of $X$ (c.f.\ \Cref{fig:lowerBounds}, left).

\begin{figure}[H]
	\centering
	\begin{subfigure}{.455\textwidth}
		\vspace*{4pt}
		\centering
		\includegraphics[width=0.92\linewidth]{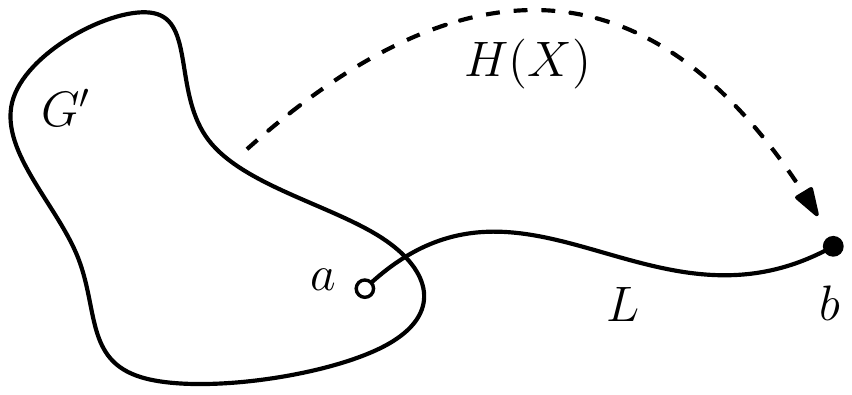}
		\label{fig:lowerBoundAPSPapproxlemma}
	\end{subfigure}%
	\begin{subfigure}{.54\textwidth}
		\vspace*{4pt}
		\centering
		\includegraphics[width=0.92\linewidth]{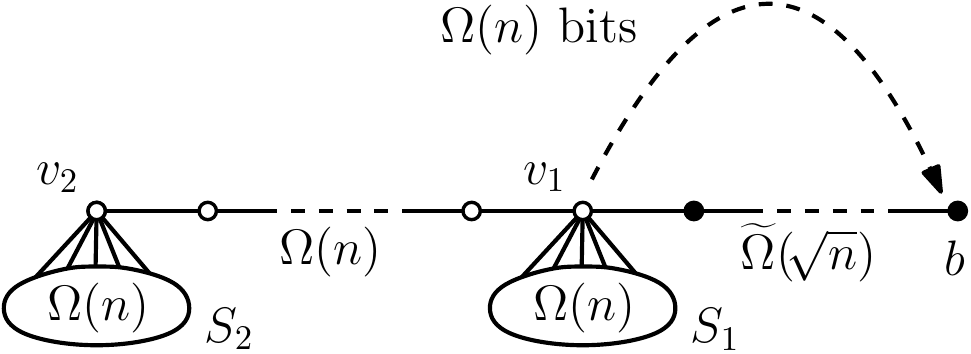}
		\label{fig:lowerBoundAPSPapprox}
	\end{subfigure}
	\caption{Construction of \Cref{lem:lowerboundlemma} (left) and of \Cref{thm:lowerBoundAPSPapprox} (right).}
	\label{fig:lowerBounds}
\end{figure}

We use \Cref{lem:lowerboundlemma} to show a lower bound of \smash{${\tilOm}\big(\!\sqrt{n}\big)$} for APSP that is robust even if we allow approximation factors up to some $\alpha \in \tilO\big(\!\sqrt{n}\big)$ (c.f.\ \Cref{thm:lowerBoundAPSPapprox}, full proof in \Cref{sec:apsp_lower_bounds}). The idea is to construct an unweighted graph consisting of a path of length $\Omega(n)$ and two node sets $S_1,S_2$ of size $\Omega(n)$ each. Let $b$ be on one end of the path. Then we attach two nodes sets $S_1$ and $S_2$ of size $\Omega(n)$ each to the path at distance approximately $L \in \tilOm(\!\sqrt{n})$ and $\Omega(n)$ from $b$ (c.f.\ \Cref{fig:lowerBounds}, right). 

Now assume an adversary ``shuffles'' the nodes in $S_1,S_2$ uniformly at random, where the state of $S_1$ will be our random variable $X$. If $b$ does not know for a node $u$ whether $u \!\in\! S_1$ or $u \!\in\! S_2$, then $b$ must assume $u \in S_2$ (recall that approximations must be lower bounded by the true distance). However, if $u \in S_1$ would be true, then this results in a larger approximation ratio than we allow. We show that in order to learn $S_1$, node $b$ needs to receive $H(X)=\Omega(n)$ bits. Choosing $L \in \tilT\big(\!\sqrt{n}\big)$ yields the claimed lower bound.

\begin{theorem}
	\label{thm:lowerBoundAPSPapprox}	
	An \smash{$\alpha$}-approximative APSP algorithm in the hybrid network model on unweighted graphs takes \smash{${\Omega}\big(\!\sqrt{n}/\log n\big)$} rounds, for any $\alpha \leq {\sqrt{nc}\cdot\log n}/{2}$, where $c \log n$ (for constant $c > 0$) is the number of messages a node can receive per round over global edges.
\end{theorem}

\subsection{Upper Bounds For Exact SSSP}

We give a sketch of the exact SSSP-algorithm, whereas the formal algorithm and proofs can be found in \Cref{sec:sssp_exact}. 
The algorithm proceeds in phases $i = 1, \ldots, \lceil2 \sqrt{\SPD} \,\rceil$. 
At the beginning of each phase two invariants are maintained: (i) every node $v$ knows the subgraph $G(v,2i)$ of $G$ induced by all nodes within hop-distance $2i$ to $v$, and (ii) every node knows its $t(i \!-\! 1)$-limited distance $d_{t(i - 1)}(s,v)$ to $s$, where \smash{$t(i) := \sum_{j = 1}^i j$} denotes the $i$-th triangular number.
Since any shortest path has length at most $\SPD$ and \smash{$t(\lceil2 \sqrt{\SPD} \,\rceil) \geq \SPD$}, after \smash{$\lceil2 \sqrt{\SPD} \,\rceil$} phases every node knows its exact distance to $s$. 
Moreover, as we consider positive edge weights, we have $\SPD \!<\! n$, thus our algorithm takes time $\tilO(\!\sqrt{n})$ in the worst case.

Maintaining invariant (i) is simple: every node sends all information it has learned about the graph so far to its neighbors for two rounds via local edges at the beginning of each phase.
Maintaining invariant (ii) is the main concern of our algorithm. 
Note that from invariant (i) a node $v$ knows $d_{i}(u,v)$ for every $u \in G(v,i)$. If $v$ would also know $d_{t(i-1)}(s,u)$ for every node $u \in G(v,i)$, then $v$ could easily determine its $t(i)$-limited distance to $s$ with the equation
\vspace*{-2mm}
\begin{equation}
\label{eq:exactSSSP}
d_{t(i)}(s,v) = \min_{\substack{u \in G(v,i)}} \Big(d_{t(i-1)}(s,u) + d_{i}(u,v) \Big).\vspace*{-1mm}
\end{equation}
Unfortunately, naively exchanging all distances $d_{t(i-1)}(s,u)$ among all pairs of nodes within $i$ hops of each other over the global network in order to compute \Cref{eq:exactSSSP} would either take too long or cause too much contention on the global network, as the neighborhood of a node could be of size $\Theta(n)$.
However, we will exploit the fact that it suffices that node $v$ learns the distance label $d_{uv} := d_{t(i-1)}(s,u) + d_i(u,v)$ from a node $u$ that minimizes \Cref{eq:exactSSSP} and safely disregard distance labels of other nodes in $G(v,i)$.

We define $T(u,i)$ as the shortest-path tree of $G(u,i)$. 
Note that $u$ and any node $v \in G(u,i)$ knows $T(u,i)$, due to invariant (i).
The goal of $u$ is to propagate the distance label $d_{uv}$ to all nodes in $v$ in $T(u,i)$ for which $u$ minimizes \Cref{eq:exactSSSP}.
To achieve that, we employ a recursive divide and conquer strategy, where each node $u$ starts with the tree $T:=T(u,i)$. In each recursion level the tree $T$ is split at a node $x$ whose removal decomposes $T$ into subtrees of size at most $|T|/2$. Note that a splitting node $x$ of $T$ always exists and can be computed locally by all $v \in G(u,i)$ (due to invariant (i)).

Let $R$ be a subtree of $T$ rooted at $x$. The root $u$ will take care of informing nodes $v \in T \setminus R$ about $d_{uv}$ in the next recursion, whereas the task of informing the nodes $v \in R$ about $d_{uv}$ is delegated to the children of $x$. For that purpose the root $u$ of $T$ informs $x$ about the distance $d_{ux} := d_{t(i-1)}(s,u) + d_i(u,x)$. Subsequently, $x$ instructs every child $c$ in $T$ to start another recursion in their respective subtree, by sending it the distance $d_{uc} := d_{ux} + w(x,c)$ via the local edge.

The difficult part is to send $d_{ux}$ to $x$ efficiently via the global network. Since $x$ might be splitting node of multiple trees, this could cause contention if done naively. We carefully resolve this by making every root $u$ of some tree $T$ that intends to send a message $d_{ux}$ to $x$ participate in an \textit{aggregation routine}\footnote{The aggregation protocol solves the following problem. Given an \textit{aggregation function} (e.g.\ MIN, MAX, SUM) and a set of source nodes that hold inputs, then some set of target nodes has to learn the result of the function applied to a subset of inputs.}, using techniques of \cite{AGG+18}. This ensures that $x$ obtains the smallest distance label $d_{ux}$, which is sufficient that $x$ eventually obtains the distance label minimizing \Cref{eq:exactSSSP} and, recursively, also the nodes in its subtree $R$. The aggregation protocol takes $\tilO(1)$ rounds and more details are given in \Cref{sec:sssp_exact}.

In the next recursion level, every child $c$ of some splitting node $x$ initiates another recursion call, where it must inform the nodes in their respective subtree of $R$ about its respective distance label. Node $c$ processes this recursion alongside all the recursions it already takes care of. We prove that the number of subtrees any node handles simultaneously is $O(\log n)$. The recursion depth is also $O(\log n)$ since the trees of \textit{all} recursions at least halve in size in each recursion level. In summary, the runtime of all recursion levels of iteration $i$ take $\tilO(1)$ rounds. Afterwards all nodes $v$ know $d_{t(i)}(s,v)$ which guarantees invariant (ii) for the next iteration $i\!+\!1$. A formal proof of the following theorem can be found in \Cref{sec:sssp_exact}.

\begin{theorem}\label{thm:exactSSSP}
	There is an algorithm that solves SSSP in time $\tilO\big(\!\sqrt{\SPD}\big)$, w.h.p.
\end{theorem}

\subsection{\boldmath Approximate SSSP in $\tilO(n^{1/3})$}
\label{sec:sssp_approx}

We summarize our $\big(1\!+\!o(1)\big)$-approximate SSSP algorithm with complexity $\tilO(n^{1/3})$. At its core, the approach relies on simulating the algorithm of Becker et al.\ \cite{BKK+17} that computes a $(1\!+\!\varepsilon)$-approximation of SSSP for the \textit{Broadcast Congested Clique Model} (BCC model, c.f.\ \Cref{def:BCCmodel}) in $\tilO(1)$ rounds. First, we compute a skeleton $S=(M,E_S)$ with \smash{$|M| = \tilT(n^{2/3})$} with edges between nodes at most $h \in \tilO(n^{1/3})$ hops apart (note that we always include the source $s \in M$). Using token dissemination (c.f.\ \Cref{sec:tokenDissemination}), we can simulate one round of the BCC model on $S$ in $\tilO(n^{1/3})$ time. This allows us to simulate the algorithm of \cite{BKK+17} on $S$, in order to $(1\!+\!\varepsilon)$-approximate SSSP on $S$ in $\tilO(n^{1/3})$ rounds in the hybrid model.

Again using token dissemination, we can make the distance approximations that we computed for pairs in $M \times\{s\}$ publicly known in $\tilO(n^{1/3})$ rounds. In another $h \in \tilO(n^{1/3})$ rounds all nodes in $V$ can do a local search to determine the distance to close nodes in $M$. After these steps every node $v \in V$ knows approximate distances $\tilde d_{su}$ between $s$ and any marked node $u \in M$ and also its own $h$-hop-limited distance $d_h(u,v)$ to any marked node $u \in M$. Then every node $v \in V$ can locally compute an approximate distance $\tilde d_{sv}$ to $s$ with the following equation:

\vspace*{-3mm}
\begin{equation}
\label{eq:approxSSSP}
\tilde d_{sv} := \min\Big(d_h(s,v), \min_{u\in M} \big(\tilde d_{su} + d_h(u,v)\big)\Big)\vspace*{-1mm}
\end{equation}

\begin{theorem}
	\label{thm:approximateSSSP}
	There is an algorithm that computes a $(1\!+\!\varepsilon)$-approximation of SSSP in $\tilO(n^{1/3} \!/ \varepsilon^{6})$ rounds.
\end{theorem}
\vspace*{-1mm}
Choosing (e.g.) $\varepsilon = \frac{1}{\log n}$ we obtain a $(1\!+\!o(1))$-approximate SSSP algorithm with complexity \smash{$\tilO(n^{1/3})$}. More details on the algorithm and the proof of \Cref{thm:approximateSSSP} are given in \Cref{sec:sssp_approx}.

\subsection{\boldmath Approximate SSSP in $\tilO(n^\varepsilon)$}
\label{sec:subpolynomial}

Finally, we present a $(\log_\alpha n)^{O(\log_\alpha n)}$-approximate SSSP algorithm that takes time $\tilO(\alpha^3)$, w.h.p., for a parameter $\alpha \geq 5$.
By setting $\alpha = n^\varepsilon$ for some $\varepsilon > 0$, we obtain a $(1/\varepsilon)^{O(1/\varepsilon)}$-approximation in time $O(n^{3\varepsilon})$, which, for example, allows to compute a constant factor approximation for any constant $\varepsilon$.
Furthermore, for $\varepsilon = \sqrt{\log \log n / \log n}$, this gives a \smash{$2^{O(\!\sqrt{\log n \log \log n})}$} approximate solution in subpolynomial time \smash{$2^{O(\!\sqrt{\log n \log \log n})}$}. We describe the algorithm from a high level and provide all details in \Cref{sec:subpolynomial}.

The main idea of the algorithm is to recursively construct a hierarchy of spanners $G_1, \ldots, G_T$, where $G_{i}$ is a spanner of the nodes in $M_i \subseteq V[G_{i-1}]$. The set $M_i$ contains each node of $G_{i-1}$ with probability $\log(n)/\alpha$ for $i=2$, and with probability $1/\alpha$ for $i\geq 3$.
The first spanner $G_1$, which contains all nodes of $G$, is constructed using only the local network by simply performing the distributed Baswana-Sen algorithm \cite{BS07} with parameter $k$ as a black box, which gives a $(2k - 1)$-spanner in time $O(k^2)$.

The construction of subsequent spanners $G_2, \ldots, G_T$ relies entirely on the global network. For $i \ge 2$, we construct $G_i$ as \emph{$h$-hop skeleton spanner} of $G_{i-1}$, which is formally defined in \Cref{def:skeletonspanner}. Intuitively, a skeleton spanner gives a good approximation of distances between nodes that are within hop-distance $h$.
We ensure that every edge of the spanner $G_i$ is learned by one endpoint of the edge in such a way that no node has to take care of more than $\tilO(\alpha)$ edges. This property (low arboricity\footnote{The arboricity of a graph is the minimum number of forests required to cover all edges.}) allows to efficiently apply the techniques of \cite{AGG+18} on $G_i$ via the global network. More specifically, we will prove that we can construct $G_{i}$ as $O(\alpha)$-hop skeleton spanner of stretch $O(\log_\alpha n)$ of the graph $G_{i-1}$ for all $i \ge 2$.

Finally, by taking the union of all the graphs $G_i$, we obtain a global spanner $H$ for the whole graph. Applying the properties of the skeleton spanners $G_i$, we show that $H$ has an $(\log_{\alpha} n)^{O(\log_{\alpha} n)}$-approximate path consisting of at most $\tilO(\alpha)$ hops for every pair of nodes $u,v\in V$. 
Thence, every node $v$ learns a good approximation $\tilde{d}(s,v)$ of $d(s,v)$ by performing a BFS from $s$ in $H$ for $\tilO(\alpha)$ rounds.
Again using techniques of \cite{AGG+18}, one round of BFS can be realized in the global network in time $\tilO(\alpha)$.
The following theorem results from careful analysis of the approximation guarantees and runtime of our recursive spanner construction.

\begin{theorem} \label{theorem:polySSSPSimple}
	The algorithm solves $(\log_\alpha n)^{O(\log_\alpha n)}$-SSSP in time $\tilO(\alpha^3)$, w.h.p.
\end{theorem}

\section{Token Dissemination}
\label{sec:tokenDissemination}

In this section we give the details of \Cref{alg:tokenDissemination} solving the $(k,\ell)$-TD problem and its subroutines. Finally we provide a full proof of \Cref{thm:tokenDissemination}.

\begin{algorithm}[H]	
	\caption{\texttt{Token-Dissemination}  \Comment{$x \in [2..k]$} }	
	\label{alg:tokenDissemination}
	\begin{algorithmic}
		\State \texttt{Token-Balancing} \Comment{\textit{redistribute tokens such that} $\ell = \tilO\big(\tfrac{k}{n}\big)$} 
		\State \texttt{Token-Multiplication} \Comment{\textit{spread \smash{$\tilO\big(\frac{n}{k}\big)$} copies of each token}} 
		\State \texttt{Token-Seeding($x$)} \Comment{\textit{seed tokens to nodes with prob.\ $1/x$}} 
		\State \texttt{Local-Dissemination($x$)} \Comment{\textit{spread seeds via local edges}}
	\end{algorithmic}			
\end{algorithm} 

We start with a technical lemma showing that if every node sends $O(\log n)$ messages to random nodes in $V$, then every nodes receives only $O(\log n)$ w.h.p.

\begin{lemma}
	\label{lem:receiveBound}
	Presume some algorithm takes at most $p(n)$ rounds for some polynomial $p$. Presume that each round, every node sends at most $\sigma = \Theta(\log n)$ messages via global edges to $\sigma$ targets in $V$ picked independently and uniformly at random. Then there is a $\rho = \Theta(\log n)$ such that for sufficiently large $n$, in every round, every node in $V$ receives at most $\rho$ messages per round w.h.p.
\end{lemma}

\begin{proof}
	Let $X_{v,r}$ be the (random) number of messages $v$ receives in round $r$. Node $v$ is targeted by some node from $V$ with probability at most $\tfrac{\sigma}{n}$ (for simplicity we assume that $v$ can send a token to itself, in reality it can just keep it). Hence we have $\mathbb{E}\big(X_{v,r}\big) \leq n \cdot \tfrac{\sigma}{n} = \sigma$. By definition we have $\sigma \geq \xi \ln n$ for some constant $\xi$ and large enough $n$. Let $c>0$ be arbitrary.  We choose \smash{$\rho \geq (1 \!+\! \frac{3c}{\xi}) \sigma$}. Then a Chernoff bound\footnote{For completeness, the variants of the Chernoff bounds that we are using in this section are given in \Cref{lem:chernoffbound}.} yields
	$$\mathbb{P}\Big(X_{v,r} \!>\! \rho \Big) \leq \mathbb{P}\Big(X_{v,r} \!>\! (1 \!+\! \tfrac{3c}{\xi})\sigma \Big) \leq \exp\Big(\!\!-\! \frac{3\xi c \ln n}{3\xi}\Big) =  \frac{1}{n^{c}}.$$
	In accordance with the union bound given in \Cref{lem:unionbound} the event \smash{\hspace*{0.1mm} $\bigcap_{\substack{\!\!\!v \in V\\ \hspace*{-0.5mm}\!\! r \leq p(n)}} \hspace*{-1mm} \Big( X_{v,r} \!\leq\! \rho \Big)$} takes place w.h.p.
\end{proof}

The following algorithm balances the number of tokens per node to $\tilO\big(\lceil k/n \rceil \big)$\footnote{The rounding brackets $\lceil\cdot\rceil$ in \smash{$\bigO\big(\lceil\tfrac{k}{n}\rceil \log n\big)$} mean that for $k \ll n$ we guarantee $\bigO(\log n)$ tokens per node.}. The set of tokens received during the execution of \Cref{alg:loadbalancing} forms the new set of tokens a node has to take care of.

\begin{algorithm}[H]	
	\caption{\texttt{Token-Balancing} \Comment{$\sigma = \Theta(\log n)$}}	
	\label{alg:loadbalancing}
	\begin{algorithmic}
		\State $T_v \gets$ initial set of tokens of this node $v$
		\For {$\bigO\big(\ell / \log n\big)$ rounds} \Comment{\textit{redistribute tokens}}
		\State $U_v \gets$ sample $\sigma$ nodes from $V$ independently and uniformly at random \Comment{\textit{$U_v$ is a multiset}}
		\State $S \gets$ select subset of size $\min(\sigma,|T_v|)$ from $T_v$
		\State $v$ sends all tokens in $S$ via global network to nodes in $U_v$ (one-to-one)
		\State $T_v \gets T_v \setminus  S$
		\EndFor
	\end{algorithmic}			
\end{algorithm}

\begin{lemma}
	\label{lem:loadbalancing}
	If each node holds at most $\ell$ tokens and there are $k$ tokens in total, then \Cref{alg:loadbalancing} redistributes all tokens in $\bigO\big(\ell / \log n\big)$ rounds such that afterwards each node holds \smash{$\bigO\big(\lceil\tfrac{k}{n}\rceil \log n\big)$} tokens w.h.p.
\end{lemma}

\begin{proof}
	Note that the loop in \Cref{alg:loadbalancing} runs sufficiently many rounds so that all nodes can transfer all of their initial tokens $T_v$. Whenever a node $v$ picks itself as target node ($v \in U_v$), it simply keeps one of its tokens and considers it as received. Due to \Cref{lem:receiveBound} no node receives more than $O(\log n)$ tokens w.h.p.
	
	If we fix some node and some token, then the node gets that token with probability $\tfrac{1}{n}$. Let $X_v$ be the number of tokens transferred to node $v \in V$. Then \smash{$\mathbb{E}(X_v) = \tfrac{k}{n}$}. Let $c>0$. We distinguish two cases. First, assume $k \geq n$. We obtain the following with a Chernoff bound:
	$$ \mathbb{P}\Big(X_v > (1 \!+\! 3c\ln n) \frac{k}{n} \Big) \leq \exp\Big(\!\!-\!\frac{3ck\ln n}{3n}\Big) \stackrel{k \geq n}{\leq} 
	\frac{1}{n^c}.$$
	This means that $X_v = \bigO\big(\tfrac{k}{n}\log n\big)$ w.h.p.\ if $k \geq n$. Second, if $k < n$ we have
	$$ 
	\mathbb{P}\Big(X_v \!>\! \big(1 \!+\! \frac{3cn\ln n}{k}\big)\frac{k}{n} \Big) \leq \exp\Big(\!\!-\!\frac{3ckn\ln n}{3kn}\Big) = \frac{1}{n^c}.$$
	Thus \smash{$X_v = \bigO\big(\!\log n\big)$} w.h.p.\ for $k < n$. Let $E_v$ be the event that \smash{$X_v \!\in\! \bigO \big(\lceil\tfrac{k}{n}\rceil \log n\big)$}. From the above we see that $E_v$ occurs w.h.p.\ in either case. By \Cref{lem:unionbound} $E_v$ takes place for every node in every round w.h.p.
\end{proof}

If $k$ is small, the next algorithm boosts the number of nodes that hold a fixed token by a factor $\tilde \Theta\big(\frac{n}{k}\big)$.

\begin{algorithm}[H]
	\caption{\texttt{Token-Multiplication}}
	\label{alg:tokenMultiplication}
	\begin{algorithmic}
		\State $T_v \gets$ initial set of tokens of this node $v$\Comment{\textit{we treat $T_v$ as a multiset}}
		\For {$\lfloor \log_2 (\frac{n}{k}) \rfloor$ phases} \Comment{\textit{runs only for $k \leq n/2$}}
		\For {each $t \in T_v$} \Comment{\textit{$|T_v| = \bigO (\log n)$ w.h.p.}}
		\State pick $u_1, u_2 \in V$ independently, uniformly at random
		\State $v$ sends a copy of $t$ to $u_1,u_2$ via global network
		\EndFor
		\State $T_v \gets$ \textit{multiset} of token-copies received in \textit{this} phase 
		\EndFor
	\end{algorithmic}
\end{algorithm}

\begin{lemma}
	\label{lem:tokenMultiplication}
	Presume that $k \leq n/2$ and each node has at most ${\bigO}(\log n)$ tokens. By invoking \Cref{alg:tokenMultiplication}, each token $t$ is copied to a random subset $V_t \subseteq V$ with \hide{$|V_t|\leq \frac{n}{k}$ and}{$|V_t| \geq \frac{n}{k\zeta \ln n}$} w.h.p., for some constant $\zeta  > 0$.  \Cref{alg:tokenMultiplication} takes \smash{$\bigO(\log n\log \frac{n}{k})$} rounds. Afterwards, we still have ${\bigO}(\log n)$ tokens per node.
\end{lemma}

\begin{proof}
	Note that if node $v$ picks itself as recipient for a token copy (which we allow), it just keeps one for the next phase. Since we choose targets randomly, no node receives more than $O(\log n)$ messages per round w.h.p., due to \Cref{lem:receiveBound}.
	
	Let $\varphi := \lfloor \log_2 (\frac{n}{k}) \rfloor$ be the number of phases of \Cref{alg:tokenMultiplication}. In each phase the total number of copies of a token $t$ in the whole network (stored locally in the variables $T_v$, $v\in V$) exactly doubles, even if multiple copies of $t$ end up at the same node, since $T_v$ is defined as a multiset. Notice that we only carry the token-copies received in the current phase over to the next phase. This means that the number of distinct nodes that hold a copy  of $t$ after $\varphi$ phases is upper bounded by $2^{\varphi}$, hence $|V_t| \leq 2^{\varphi} \leq \frac{n}{k}$. 
	
	For the lower bound of $V_t$ we show that in every phase, the multiset $T_v$ contains at most $\bigO(\log n)$ copies of tokens w.h.p. We emphasize that for the sake of this proof we distinguish \textit{token-copies} when counting them, even if they originate from the same token. 
	Initially the claim is true due to the presumption. After that, \Cref{alg:tokenMultiplication} distributes all created token-copies uniformly at random.
	
	In every phase, given \textit{one} specific copy of a token $t$, a given node receives that copy with probability exactly $\frac{1}{n}$. Hence, the expected number of copies of any token $t$ is at most \smash{$\mathbb{E}(|T_v|) \leq k \!\cdot\! |V_t| \!\cdot\! \frac{1}{n} \leq 1$} (recall $|V_t| \leq \frac{n}{k}$). For an arbitrary constant $c > 0$ we obtain \smash{$\mathbb{P}\big(|T_v| > 1\!+\!3c\log n \big) \leq 
	\frac{1}{n^c}$} with a Chernoff bound (\Cref{lem:chernoffbound}), i.e. $|T_v| = \bigO(\log n)$. With a union bound (\Cref{lem:unionbound}), $|T_v| = \bigO(\log n)$ holds w.h.p.\ for all nodes in all rounds of \Cref{alg:tokenMultiplication}. Since $|T_v|$ is also the time complexity of a single phase (c.f.\ \Cref{alg:tokenMultiplication}), this proves the running time of \Cref{alg:tokenMultiplication}.
	
	Let constant $\zeta  \!>\! 0$ be such that \smash{$|T_v| \leq \tfrac{\zeta}{2} \ln n$} for almost all $n$. This means any node holds at most \smash{$\tfrac{\zeta}{2} \ln n$} token-copies, when we treat copies of the same token $t$ as \textit{different} copies ($T_v$ is a multiset). Conversely, the size of the \textit{set} $V_t$ represents the overall number of copies of a token $t$ in the network, in case we count multiple copies of the same token $t$ at the same node as \textit{one}. Therefore, the upper bound of $2^\varphi$ for $|V_t|$ differs from the lower bound by a factor of at most \smash{$\tfrac{\zeta}{2} \ln n$}. We obtain
	$$|V_t| \!\geq\! \frac{2^{\varphi}}{\max_{v\in V}|T_v|} \!\geq\! \frac{2^{\varphi}}{\zeta/2  \ln n} \!=\! \frac{2^{\varphi+1}}{\zeta  \ln n} \!\geq\! \frac{2^{\log_2 (n/k)}}{\zeta  \ln n} \!=\! \frac{n}{k \zeta\ln n}.$$
	
	We already established the fact that the number of tokens-copies per node is at most $|T_v| = \bigO(\log n)$. The same is obviously true for the number of distinct tokens per node thus we have $\bigO(\log n)$ tokens per node w.h.p., after the execution of \Cref{alg:tokenMultiplication}.
\end{proof}

The goal of the next algorithm is to seed each token to roughly $\frac{n}{x}$ random nodes. For $k = \tilOm(n)$ we can afford to sample targets with probability $\frac{1}{x}$ for each token and send them via the global network within our target runtime (which is \smash{$\bigO\big({\ell \cdot \min(k,n)}/{x }\big)$}). For small $k$ we decrease the sampling rate to $\tilO\big(\frac{k}{nx}\big)$. The algorithm still works because now $|V_t| = \tilT(\frac{n}{k})$ nodes are helping to seed token $t$ (c.f.\ \Cref{lem:tokenMultiplication}).

\begin{algorithm}[H]
	\caption{\texttt{Token-Seeding($x$)} \Comment{$x \in [2..k]$, $\sigma \!\in\! \Theta(\log n)$, \textit{$\zeta>0$ is from \Cref{lem:tokenMultiplication}}}}
	\label{alg:tokenSeeding}
	\begin{algorithmic}
		\State $T_v \gets $ initial set of tokens of this node $v$ 
		\For {$t \in T_v$}
		\If {$k \geq \tfrac{n}{2\zeta \ln n}$}  $S_t \gets $ sample each $u \in V$ with probability \smash{$\tfrac{1}{x}$} \Comment{\textit{full sampling rate}}
		\Else $\;\, S_t \gets $ sample each $u \in V$ with probability \smash{$ \frac{k}{nx} \!\cdot\! 2\zeta \ln n$} \Comment{\textit{reduced sampling rate}}
		\EndIf
		\While{$S_t \neq \emptyset$} \Comment{\textit{seeding token $t$}}
		\State $S \gets $ uniformly random subset of $S_t$ of size $\min(\sigma,|S_t|)$
		\State $v$ sends $t$ via global network to all \smash{$u \in S \setminus \{v\}$}
		\State \smash{$S_t \gets S_t \setminus S$}
		\EndWhile
		\EndFor		
	\end{algorithmic}
\end{algorithm}

\begin{lemma}
	\label{lem:tokenSeeding}
	If each node has initially at most $\ell$ tokens, then w.h.p.\ after Algorithms \ref{alg:tokenMultiplication} and \ref{alg:tokenSeeding}, each node knows any given token with probability at least ${1}/{x}$. \Cref{alg:tokenSeeding} takes \smash{$\bigO\big({\ell \cdot \min(k,n)}/{x }\big)$} rounds.
\end{lemma}

\begin{proof}
	Each node sends a token to at most $\sigma$ uniformly random nodes (a priori, the probability of being selected as a target by a fixed node in some fixed round is equal for every node). As before we invoke \Cref{lem:receiveBound} to argue that w.h.p., no node $v$ receives more than $O(\log n)$ messages per round. 
	
	Next we show that the sampled sets $S_t$ are not too large. First consider the case \smash{$k \!\geq\! \tfrac{n}{2\zeta \ln n}$}. In this case we sample with probability $1/x$ thus \smash{$\mathbb{E}(|S_t|) = \frac{n}{x}$} and with a standard Chernoff bound we have $|S_t| = \bigO(n \log n / x)$ w.h.p.
	Now consider \smash{$k \!<\! \tfrac{n}{2\zeta \ln n}$}. In this case the sample probability is reduced to \smash{$\tfrac{2k \zeta \ln n}{xn}$}. For the expectation we get \smash{$\mathbb{E}(|S_t|) = \tfrac{2k \zeta \ln n}{x}$}. For some $c > 0$ and with a Chernoff bound we obtain
	$$ \mathbb{P}\Big(|S_t| \!>\! \big(1\!+\!\frac{3c}{2\zeta}\big)\!\cdot\!\frac{2k\zeta \ln n}{x} \Big) \leq \exp\Big(\!\!-\!\frac{k c \ln n}{x}\Big) 
	\stackrel{x \leq k}{\leq} 
	\frac{1}{n^c},$$
	thus the event $|S_t| = \bigO\big({k \log n}/{x}\big)$ occurs w.h.p. Combining both cases we have $|S_t| = \bigO\big({\min(k,n) \log n}/{x}\big)$ w.h.p. In accordance with \Cref{lem:unionbound} this is true for every node $v \in V$, every token $t \in T_v$ and in every round of \Cref{alg:tokenSeeding}. Now we are able to compute the time complexity: Sending every token $t \in T_v$ to each of the sampled nodes in $S_t$ takes  $\bigO\big({|T_v| \!\cdot\! |S_t|}/{\sigma}\big) = \bigO\big({\ell \cdot \min(k,n)}/{x }\big)$ rounds.
	
	It remains to be shown that in case \smash{$k \!<\! \tfrac{n}{2\zeta \ln n}$} we can still guarantee that each node obtains a given token with probability at least $1/x$ even though we sample with reduced probability \smash{$\tfrac{2k \zeta \ln n}{xn}$}. In that case, we know from \Cref{lem:tokenMultiplication} that during a run of \Cref{alg:tokenMultiplication} each token $t$ is copied to a random subset \smash{$V_t \subseteq V$} of nodes of size \smash{$|V_t| \geq \frac{n}{k \zeta \ln n}$}. In \Cref{alg:tokenSeeding} all nodes of $V_t$ take part in seeding $t$. 	
	Let \smash{$p := \tfrac{2k \zeta \ln n}{xn}$} and let \smash{$q := \frac{n}{k \zeta \ln n}$}. The probability that fixed node receives a fixed token is at least $1\!-\!(1\!-\!p)^q$. We will show that $1\!-\!(1\!-\!p)^q \geq \frac{pq}{2} = \frac{1}{x}$ for $0 \leq pq \leq 1$ (which holds for $x \geq 2$). We have 
$$(1-p)^q = \Big(\big(1-\frac{1}{1/p}\big)^{1/p}\Big)^{pq} \leq e^{-pq} \leq 1 \!-\! \frac{pq}{2}.$$
The last inequality holds since for $pq=0$ we have equality; $e^{-1} < \frac{1}{2}$ for $pq=1$ and since $e^{x}$ is convex.	
\end{proof}

\begin{algorithm}[H]	
	\caption{\texttt{Local-Dissemination($x$)} \Comment{$x \in [2..k]$}}
	\label{alg:localDissemination}
	\begin{algorithmic}
		\State $T_v \gets $ all tokens that $v$ learned during \Cref{alg:tokenSeeding}: \texttt{Token-Seeding}
		\For {$r = \bigO(x \log n)$ rounds} \Comment{\textit{tokens travel $\bigO(x \log n)$ hops via local edges}}
		\State $v$ sends $T_v$ to all its neighbors in $G$ via local edges
		\State $T_v \gets $ all tokens that $v$ learned for the first time in the last round
		\EndFor
	\end{algorithmic}			
\end{algorithm}

\begin{lemma}
	\label{lem:localDissemination}
	Let $G$ be a connected. If for all tokens $t$, any given node knows $t$ with probability at least $\tfrac{1}{x}$, then after \Cref{alg:localDissemination} is performed, all nodes know all tokens w.h.p.\ after $\bigO(x \log n)$ rounds.
\end{lemma}

\begin{proof}
	Let $r = cx \ln n$. After $r$ rounds in the loop, all nodes know all tokens within their $r$-hop neighborhood $N_r(v)$. For a certain token $t$ let $X_{v,t}$ be the number of nodes $v \in N_r(v)$ that know $t$. 
	
	Since $G$ is connected we have $|N_r(v)| \geq r$ thus {$\mathbb{E}(X_{v,t}) = {|N_r(v)|}/{x} \geq {r}/{x} = c \ln n$}. We apply another Chernoff bound	
	$$\mathbb{P}\Big(X_{v,t} < 1 \Big) \leq \mathbb{P}\Big(X_{v,t} < \frac{\mathbb{E}(X_{v,t})}{2} \Big) \leq \exp\Big(\!\!-\! \frac{c\ln n}{8}\Big) \stackrel{c' \coloneqq c/8}{\leq} \frac{1}{n^{c'}}.$$
	Hence, w.h.p.\ $v$ has at least one node in its $r$-hop neighborhood that knows $t$. Since $k \leq n^2$ the event \smash{$\;\bigcap_{\substack{v \in V\\ \!\!\!\!\!\text{ tokens }t}} \!\!\big(X_{v,t} \!\geq\!1 \big)$} occurs w.h.p., due to \Cref{lem:unionbound}.
\end{proof}

We stitch together the results we have shown so far to prove \Cref{thm:tokenDissemination}. 

\begin{proof}[Proof of \Cref{thm:tokenDissemination}]
	Let $\ell_{\text{init}}$ be the initial maximum number of tokens per node. First we execute \Cref{alg:loadbalancing}: \texttt{Token-Balancing} which takes $\bigO(\ell_{\text{init}})$ rounds. Afterwards we have \smash{$\ell = \bigO\big(\lceil\frac{k}{n}\rceil \log n\big)$} tokens per node w.h.p., in accordance with \Cref{lem:loadbalancing}.
	
	In case $k \leq n/2$ we run at least one phase of \Cref{alg:tokenMultiplication}: \texttt{Token-Multiplication}. Due to $k \leq n/2$ we have $\ell = \bigO(\log n)$ and therefore \Cref{alg:tokenMultiplication} takes $\tilO(1)$ rounds according to \Cref{lem:tokenMultiplication}. We have shown that the condition $\ell = \bigO(\log n)$ is preserved by \Cref{alg:tokenMultiplication}.
	
	In accordance with \Cref{lem:tokenSeeding}, \Cref{alg:tokenSeeding}: \texttt{Token-Seeding} takes \smash{$\bigO\big(\ell \cdot {\min(k,n)}/{x}\big)$} rounds. The maximum number of tokens is $\ell = \bigO\big(\lceil\frac{k}{n}\rceil \log n\big)$. Since $\bigO\big(\!\min(k,n) \lceil\frac{k}{n}\rceil \big) = \bigO(k)$ the time complexity of \Cref{alg:tokenSeeding} reduces to \smash{$\tilO\big(\tfrac{k}{x}\big)$}.
	
	After \texttt{Token-Seeding} has terminated, the premise of \Cref{lem:localDissemination} is fulfilled. Thus \Cref{alg:localDissemination}: \texttt{Local-Dissemination} solves the $(k,\ell)$-Token-Dissemination Problem in $\tilO(x)$ rounds.
	
	The total number of rounds of \Cref{alg:tokenDissemination}: \texttt{Token-Dissemination} is \smash{$\tilO(x) + \tilO\big(\tfrac{k}{x}\big) + \bigO(\ell_{\text{init}})$}. This is optimized for $x = \Theta\big(\!\sqrt{k}\big)$, which results in the overall time complexity \smash{$\tilO\big(\!\sqrt{k} + \ell_{\text{init}}\big)$}.
\end{proof}

Finally, we show that we do not need the full power provided by the local network (\LOCAL model).

\begin{lemma}
	\label{lem:congestion_token_dissem}
	\Cref{alg:tokenDissemination} works in the same time \smash{$\tilO\big(\!\sqrt{k} + \ell\big)$} for local capacity $\lambda = \Theta\big(\!\sqrt{k}\big)$.
\end{lemma}

\begin{proof}
	Note that the only instance where local edges are used is in the sub-procedure \Cref{alg:localDissemination}, where each node distributes the tokens it learns for the first time via its local edges. We aim to apply \Cref{lem:balance_congestion} in \Cref{apx:balance_congestion}. For each token $t$ let algorithm $\mathcal A_t$ be responsible for disseminating token $t$. That is, each $\mathcal A_t$ runs \Cref{alg:tokenDissemination} on each node but restricted to token $t$, i.e.\ it disseminates only token $t$ as soon as it learns $t$ for the first time (or knows it at the start of the algorithm). It is clear that running all $\mathcal A_t$ in parallel has the same outcome as \Cref{alg:localDissemination}.
	
	As any $\mathcal A_t$ sends $t$ at most twice over each edge (once from each endpoint) and since the algorithms $\mathcal A_t$ are obviously independent, they are simple. Each token has to travel at least \smash{$\Theta\big(\!\sqrt{k}\big)$} hops in $G$, thus we have \smash{$D = \Theta\big(\!\sqrt{k}\big)$} for the dilation. The maximum number of messages send over one edge in one round is $C=k$. Given a local capacity of $\lambda$, we can execute all $\mathcal A_t$ in time $O(C/\lambda + D + \log n) = \tilO(k/\lambda + \sqrt{k})$ using the method of \Cref{lem:balance_congestion}. Thus we can restrict can restrict ourselves to some $\lambda = \Theta\big(\!\sqrt{k}\big)$ and still maintain the same overall running time given in \Cref{thm:tokenDissemination}.
\end{proof}

  \section{All Pairs Shortest Paths}
\label{sec:apsp}
This section focuses on the APSP problem. We first show how to solve APSP exactly. Second, we show that a significant improvement in the time complexity is possible, if we restrict ourselves to approximations. Finally, we prove that the running times of our approximate algorithms are tight up to $\polylog n$ factors.

\subsection{Upper Bounds for Exact APSP}
\label{sec:apsp_exact}

In the following we present \Cref{alg:exactAPSP} and its subroutines and show their properties. Subsequently we prove \Cref{thm:APSP}.

\begin{algorithm}[H]	
	\caption{\texttt{Exact-APSP}  \Comment{$x \in [1..n]$} }	
	\label{alg:exactAPSP}
	\begin{algorithmic}
		\State \texttt{Construct-Skeleton($x$)} \Comment{\textit{construct subgraph $S=(M,E_S)$}}
		\State \texttt{Transmit-Skeleton} \Comment{\textit{transmit $S$ to all nodes via $G'$}}
		\State \texttt{Transmit-Distances} \Comment{\textit{transmit distances to close nodes in $S$}}
		\State compute APSP distances locally with \Cref{eq:exactAPSP}
	\end{algorithmic}			
\end{algorithm}

\begin{equation*}
D^G_{uv} = \min \Big(d_h(u,v), \min_{u'\!,v' \in M} \big( D'_{uu'} \!+ D^S_{u'v'} \!+ D'_{vv'}\big)\Big).\tag{\ref{eq:exactAPSP}}
\end{equation*}

First we construct the skeleton $S$ by sampling its nodes and then determining its edges via exploration on the local network. As byproduct of the latter, all nodes learn their $h$-hop neighborhood in $G$.

\begin{algorithm}[H]	
	\caption{\texttt{Construct-Skeleton($x$)}  \Comment{$x \!\in\! [1..n], h \!:=\! \xi x \ln n$}}
	\label{alg:constructskeleton}
	\begin{algorithmic}
		\State $v$ is marked with prob.\ $\tfrac{1}{x}$ \Comment{\textit{marked nodes form skeleton nodes $M$}}
		\For{$h$ rounds}\Comment{\textit{learn proximate nodes in $M$}}
		\State $v$ sends knowledge about $G$ it learned last round via local edges \Comment{\textit{initially $v$ sends its incident edges}}
		\EndFor
	\end{algorithmic}			
\end{algorithm}

\begin{fact}
	\label{fct:constructskeleton}
	The size of $M$ is \smash{$\tilO(n/x)$} w.h.p. Let $h := \xi x \ln n$. \Cref{alg:constructskeleton} establishes a weighted graph $S = (M,E_S)$ among the set of marked nodes $M$ in ${O}(x \log n)$ rounds, whereas we define $E_S \coloneqq \{ \{u,v\} \!\mid\! u,v\!\in\! M, \text{hop}(u,v) \!\leq\! h\}$. The weight of $\{u,v\} \in E_S$ is defined as $d_h(u,v)$.
	After the subroutine, all nodes $v \in V$ know all $u \in V$ that are within $h$ hops as well as the distances $d_h(u,v)$. Specifically, this means that all marked nodes know their neighbors in $S$ and the distances of the incident edges in $E_S$.
\end{fact}

The following lemma shows that for nodes at sufficient hop-distance $\tilOm(h)$, there is a marked node every $\tilO(h)$ hops on some shortest path between those nodes.

\begin{lemma}
	\label{lem:longPathsM}
	Let $M$ be a subset of $V$ created by marking each of the nodes of $V$ with probability at least $\frac{1}{x}$. Then there is a constant $\xi \!>\! 0$, such that for any $u,v \!\in\! V$ with $hop(u,v) \!\geq\! \xi x \ln n$, there is at least one shortest path $P$ from $u$ to $v$, such that any sub-path $Q$ of $P$ with at least $\xi x \ln n$ nodes contains a node in $M$ w.h.p.
\end{lemma}

\begin{proof}
	Let $u,v \in V$ with $hop(u,v) \!\geq\! \xi x \ln n$. Fix a shortest $u$-$v$-path $P_{u,v}$ and let $Q$ be a sub-path of $P_{u,v}$ with at least $\xi x \ln n$ nodes. Let $X_{u,v}$ be the random number of marked nodes on $Q$. Then we have \smash{$\mathbb{E}(X_{u,v}) \geq \frac{|Q|}{x} \geq {\xi \ln n}$}. Let $c > 0$ be arbitrary. We use a Chernoff bound:
	$$ \mathbb{P}\Big( X_{u,v} < \frac{\xi \ln n}{2}\Big) \leq \exp\Big(\!\!-\! \frac{\xi \ln n}{8}\Big) \stackrel{\xi \geq 8c}{\leq} \frac{1}{n^c}.$$
	Thus we have $X_{u,v} \geq 1$ w.h.p.\ for constant \smash{$\xi \geq \max(8c, 2 / \ln n)$}. Therefore the claim holds w.h.p.\ for the pair $u,v$. We claim that w.h.p.\ the event $X_{u,v} \geq 1$ occurs for all pairs $u,v \in V$ and for all sub-paths $Q$ of $P_{u,v}$ longer than $\xi x \ln n$ hops, for at least one shortest path $P_{u,v}$ from $u$ to $v$. There are at most \smash{$n^2$} many pairs $u,v \in V$. Moreover we can select at most $n$ sub-paths $Q$ of $P$ that do not fully contain any other selected sub-path. Hence the claim follows with the union bound given in \Cref{lem:unionbound}.
\end{proof}

Next we make the skeleton publicly known via token dissemination.

\begin{algorithm}[H]	
	\caption{\texttt{Transmit-Skeleton} \Comment{$h = \xi x \ln n$}}
	\label{alg:transmitskeleton}
	\begin{algorithmic}			
		\If{$v$ is marked}
		\State for $\{u,v\} \!\in\! E_S$ create token \smash{$t_{u,v} = \langle I\!D(u),I\!D(v),d_h(u,v)\rangle$}
		\EndIf
		\State \texttt{Token-Dissemination} \Comment{\textit{dissem.\ $E_S$ and weights of $E_S$}}
	\end{algorithmic}			
\end{algorithm}

\begin{lemma}
	\label{lem:distancesSkeleton}
	After \Cref{alg:transmitskeleton}, w.h.p.\ every node knows the skeleton $S$ and has  sufficient information to locally compute a distance matrix $D^S$ of $S$, with $D^S_{uv} = D^G_{uv}$ for all $u,v \in M$ (where $D^G$ denotes the true distance matrix of $G$), if $h = \xi x \ln n$ for appropriately chosen constant $\xi$. \Cref{alg:transmitskeleton} takes $\tilO\big(\frac{n}{x}\big)$ rounds.
\end{lemma}

\begin{proof}
	First we point out that every marked node $v \in M$ knows $d_h(u,v)$ for all $u \in V$ due to Fact \ref{fct:constructskeleton} (recall that we set $d_h(u,v) := \infty$ if $hop(u,v)>h$) and is thus able to create the tokens described in the algorithm. Each marked node creates at most $\ell = |M|$ tokens of size ${O}(\log n)$ (recall that weights are polynomially bounded in $n$). The total number of created tokens is at most $k = |M|^2$. By \Cref{thm:tokenDissemination}, \Cref{alg:transmitskeleton} takes $\tilO\big(|M|\big) = \tilO(n/x)$ rounds. After the token dissemination \textit{every} node knows \textit{every} edge $\{u,v\} \in E_S$ as well as its weight, defined as $h$-limited distance $d_h(u,v)$ (in Fact \ref{fct:constructskeleton}).
	
	Let $h := \xi x \ln n$ (where $\xi$ is the constant from \Cref{lem:longPathsM}) and let $u,v \in M$. If there is a shortest $u$-$v$-path $P$ with $|P| \leq h$, then obviously the weight $d_h(u,v)$ of the skeleton edge $\{u,v\} \in E_S$ equals $D^G_{uv}$ (let us denote this fact with (i)). Otherwise $|P| > h$ for any shortest $u$-$v$ path $P$. Then \Cref{lem:longPathsM} implies that w.h.p., within every $h$ hops of $P$ there must be at least one marked node (we denote this fact with (ii)). This entails that $S$ is connected if $G$ is connected; and $G$ is connected by definition (let this fact be (iii)).
	
	From (i),(ii) and (iii) we deduce that every node can compute $D^S$ by locally solving APSP on $S$.
\end{proof}

It remains to transmit the distances between skeleton nodes and non-skeleton nodes.

\begin{algorithm}[H]	
	\caption{\texttt{Transmit-Distances} \Comment{$h = \xi x \ln n$}}
	\label{alg:transmitdist}
	\begin{algorithmic}			
		\If{$v$ not marked} 
		\State for each $u \!\in\! M$ create token \smash{$t_{u,v}=\langle ID(u),ID(v),d_h(u,v)\rangle$}
		\EndIf
		\State \texttt{Token-Dissemination} \Comment{\textit{dissem.\ all $h$-limited dist.\ $d_h(u,v)$}}
	\end{algorithmic}			
\end{algorithm}

\begin{fact}
	\label{fct:transmitdist}
	\Cref{alg:transmitdist} disseminates $d_h(u,v)$ for all $u \in M$ and $v \in V \!\setminus\! M$ to all nodes in the network (recall that we define $d_h(u,v) \!:=\! \infty$ if $hop(u,v)>h$). The $h$-limited distances $d_h(u,v)$ are known to $v$ due to Fact \ref{fct:constructskeleton}. Each node creates at most $\ell = |M|$ tokens (of size $O (\log n)$ bits), thus there are at most $k = n|M|$ tokens in total. Due to \Cref{thm:tokenDissemination}, \Cref{alg:transmitdist} takes \smash{$\tilO\big(\!\sqrt{n|M|}\,\big)$} rounds. Since \smash{$|M| \in \tilO\big(\frac{n}{x}\big)$} w.h.p., this translates into a running time of \smash{$\tilO\big(n/\!\sqrt{x}\,\big)$} rounds.
\end{fact}

\begin{proof}[Proof of \Cref{thm:APSP}]
	
	After the first two subroutines of \Cref{alg:constructskeleton}, due to \Cref{lem:distancesSkeleton}, every node knows $S$ and can locally compute the distance Matrix $D^S$ among all nodes in $M$. 
	Additionally, based on the $h$-limited distances disseminated by \Cref{alg:transmitdist} as described in Fact \ref{fct:transmitdist}, every node can locally compute the matrix $$D' \coloneqq \big(d_h(u,v)\big)_{u \in V\setminus M, v \in M}.$$ 
	
	Let $u,v \in V$. If there exists a shortest $u$-$v$-path that has at most $h$ hops, then $d(u,v) = d_h(u,v)$ which both $u$ and $v$ already know due to the local exploration conducted in \Cref{alg:constructskeleton} (c.f. Fact \ref{fct:constructskeleton}). Otherwise, we infer from \Cref{lem:longPathsM} that w.h.p., there is shortest $u$-$v$-path $P$ with two marked nodes $u',v' \in M$ with $hop(u,u'), hop(v,v') \leq h$ (possibly $u'=v'$). We deduce
	$$d(u,v) = d_h(u,u') + d(u',v') + d_h(v',v) \stackrel{(*)}{=} D'_{uu'} \!+ D^S_{u'v'} \!+ D'_{vv'}.$$
	Where $(*)$ is due to \Cref{lem:distancesSkeleton}. Hence every node can locally compute the complete distance matrix $D^G$ of $G$ as follows (we set $D'_{uv} = 0$ if $u,v \in M$):
	$$D^G_{uv} = \min \Big(d_h(u,v), \min_{u'\!,v' \in M} \big( D'_{uu'} \!+ D^S_{u'v'} \!+ D'_{vv'}\big)\Big).$$
	The total running time is $\tilO(x)+ \tilO\big(n/\!\sqrt{x}\big)$ due to Fact \ref{fct:constructskeleton}, \Cref{lem:distancesSkeleton} and Fact \ref{fct:transmitdist}. This is optimized for \smash{$x \in \Theta\big( n^{2/3}\big)$}.
\end{proof}

Finally, we show that local capacity $\lambda = \Theta(n^{4/3})$ suffices to solve APSP exactly in the claimed time.

\begin{lemma}
	\label{lem:exactAPSPcongestion}
	\Cref{alg:exactAPSP} works in the same time $\tilO(n^{2/3})$ for local capacity $\lambda = \Theta(n)$.
\end{lemma}

\begin{proof}
	Note that we require the local edges in subroutine \Cref{alg:constructskeleton}: \texttt{Construct-Skeleton($x$)} only to learn the $h$-limited distances to all nodes in an $h$-hop neighborhood as well as all marked nodes in the $h$-hop neighborhood in $\tilO(h) = \tilO(x)$ rounds. We can learn the former by running the distributed Bellman-Ford algorithm for APSP for $h$ rounds, which requires $\lambda = 2n$ (as we show in \Cref{lem:congestion_bellman_ford}, for completeness). The information whether a node is marked or not can picky-back on the messages of the Bellman-Ford algorithm without producing additional congestion (messages are still of size $O(\log n)$).
	
	Moreover we use local edges implicitly in \Cref{alg:transmitskeleton}: \texttt{Transmit-Skeleton} and \Cref{alg:transmitdist}: \texttt{Transmit-Distances}, where we call the token dissemination subroutine. Most congestion on edges is caused by the latter, where we have to disseminate $n|M| = \tilO(n^2/x) = \tilO\big(n^{4/3}\big)$ tokens in time $\tilO\big(n^{2/3}\big)$. In \Cref{lem:congestion_token_dissem} we show that local capacity $\lambda = \tilT(n^{2/3})$ suffices for this.	
\end{proof}

\subsection{Upper Bounds for Approximate APSP}
\label{sec:apsp_approx}

Besides slightly adapted procedures \texttt{Transmit-Closest} and \texttt{Construct-Skeleton'}\hspace*{-1mm}, \Cref{alg:approxAPSP} uses the same subroutines as \Cref{alg:exactAPSP} to construct and disseminate the skeleton $S$ and then determine its edges with a local search via the physical edges. In the following we briefly explain the (minor) changes of the subroutines of \Cref{alg:approxAPSP}. Subsequently and more importantly, we prove that the approximate distance matrix $\tilde{D}^G_{uv}$ in fact meets the claimed properties.

\begin{algorithm}[H]	
	\caption{\texttt{Approximative-APSP}  \Comment{$x \in [1..n]$} }	
	\label{alg:approxAPSP}
	\begin{algorithmic}
		\State \texttt{Construct-Skeleton'\!($x$)} \Comment{\textit{construct subgraph $S=(M,E_S)$}}
		\State \texttt{Transmit-Skeleton} \Comment{\textit{transmit $S$ to all nodes via $G'$}}
		\State \texttt{Transmit-Closest} \Comment{\textit{transmit distances to close nodes in $S$}}
		\State approximate APSP distances locally with \Cref{eq:approxAPSP}
	\end{algorithmic}			
\end{algorithm}

\begin{equation*}
\tilde{D}^G_{uv} = \min\Big(d_m(u,v), \min_{u'\in M} \big(d_h(u,u') + D^S_{u'v'}\big) + d_{vv'}\Big). \tag{\ref{eq:approxAPSP}}
\end{equation*}

As a slight adaption over the exact variant, we conduct a local exploration (\Cref{alg:constructskeletonprime}) up to hop-distance $m = \max\!\big(h,\frac{n}{h}\big)$ (instead of $h$). This allows us to use the same algorithm to compute a 3-approximation for the weighted case in \smash{$\tilO\big(\!\sqrt{n}\big)$} rounds and a $(1\!+\!\varepsilon)$-approximation in \smash{$\tilO\big(\!\sqrt{n /\varepsilon}\big)$} rounds for the unweighted case. For the latter we prove a slightly more general variant, where we get a $(1\!+\!\varepsilon)$-approximation for weighted graphs in \smash{$\tilO\big(\!\sqrt{n W/\varepsilon}\big)$} rounds, where $W := W_{\text{max}}/W_{\text{min}}$ is the (potentially large) ratio of maximum to minimum weight. The unweighted case $W=1$ is a direct corollary.

\begin{algorithm}[H]	
	\caption{\texttt{Construct-Skeleton'}($x$)  \Comment{$x \!\in\! [1..n], h \!:=\! \xi x \ln n$}}
	\label{alg:constructskeletonprime}
	\begin{algorithmic}
		\State $v$ is marked with prob.\ $\tfrac{1}{x}$ \Comment{\textit{marked nodes form nodes $M$ of skeleton}}
		\For{$m = \max\!\big(h,\frac{n}{h}\big)$ rounds}\Comment{\textit{learn proximate nodes in $M$}}
		\State $v$ sends knowledge about $G$ it learned last round via local edges \Comment{\textit{initially $v$ sends its incident edges}}
		\EndFor
	\end{algorithmic}			
\end{algorithm}

\begin{fact}
	\label{fct:constructskeletonprime}
	By performing \Cref{alg:constructskeletonprime} all nodes learn all information described in Fact \ref{fct:constructskeleton}. Additionally all nodes $u \in V$ learn the $m$-limited distances $d_{m}(u,v)$ to all nodes $v \in V$, where $m := \max(h,n/h)$. The running time is $O(\max(h,n/h)) = \tilO(\max(x,n/x))$ (recall $h \!=\! \xi x \ln n$).
\end{fact}

\begin{algorithm}[H]	
	\caption{\texttt{Transmit-Closest} \Comment{$h \!=\! \xi x \ln n$}}
	\label{alg:transmitclosest}
	\begin{algorithmic}			
		\If{$v$ not marked}
		\State $v' \gets \argmin_{w \in M} d_h(v,w)$ \Comment{\textit{node $v' \in M$ closest to $v$}}
		\State $d_{vv'} \gets d_h(v,v')$ \Comment{\textit{distance from $v$ to closest node $v' \in M$}}
		\State create token \smash{$t_{v'\!,v}=\langle ID(v),ID(v'),d_{vv'}\rangle$}
		\EndIf
		\State \texttt{Token-Dissemination} \Comment{\textit{disseminate distances $d_{vv'}$}}
	\end{algorithmic}
\end{algorithm}

\begin{fact}
	\label{fct:transmitclosest}
	Through \Cref{alg:transmitclosest} all nodes learn which node $v' \in M$ is closest to any given $v \in V\setminus M$ as well as the distance $d_{vv'}$ between $v$ and $v'$. Since each node creates only one token, there are at most $n$ tokens in total. Due to \Cref{thm:tokenDissemination}, \Cref{alg:transmitclosest} takes \smash{$\tilO\big(\!\sqrt{n}\big)$} rounds w.h.p.
\end{fact}

In order to prove the approximation ratios claimed at the beginning of this section, we give a number of notations and we call upon the reader to consult \Cref{fig:apsp_approx} for a graphic overview. Assume that \Cref{alg:approxAPSP} has terminated. Let $u,v \in V$ be a pair of nodes, for which all shortest $u$-$v$-paths have more than $h$ hops. Let $P$ be the shortest $u$-$v$-path that minimizes $hop(w,v)$, where $w \in M$ is the marked node on $P$ that is closest to $v$. We denote the sub-path of $P$ from $w$ to $v$ with $Q$. From \Cref{lem:longPathsM} we know that $|Q| \leq h$ w.h.p. 

Furthermore let $v' \in M$ be the node that minimizes $d_h(v,v')$, which corresponds to $d_{vv'}$ and is known to all nodes due to Fact \ref{fct:transmitclosest}. Let $O$ be a shortest $u$-$v'$-path. In case there are several, let $O$ be the shortest $u$-$v'$-path that has a marked node on every sub-path with at least $h$ hops (which exists w.h.p.\ due to \Cref{lem:longPathsM}). Let $u'$ be the marked node on $O$ closest to $u$, i.e., $hop(u,u') \leq h$. Note that $u' = v'$ is possible. Let $R$ be the path compounded of $O$ and a shortest path from $v'$ to $v$. The next Lemma shows that for $u$, $v$ with $hop(u,v) \geq h$ the result of \Cref{eq:approxAPSP} is $w(R)$ w.h.p.

\begin{figure}[H]
	\centering
	\includegraphics[width=0.6\linewidth]{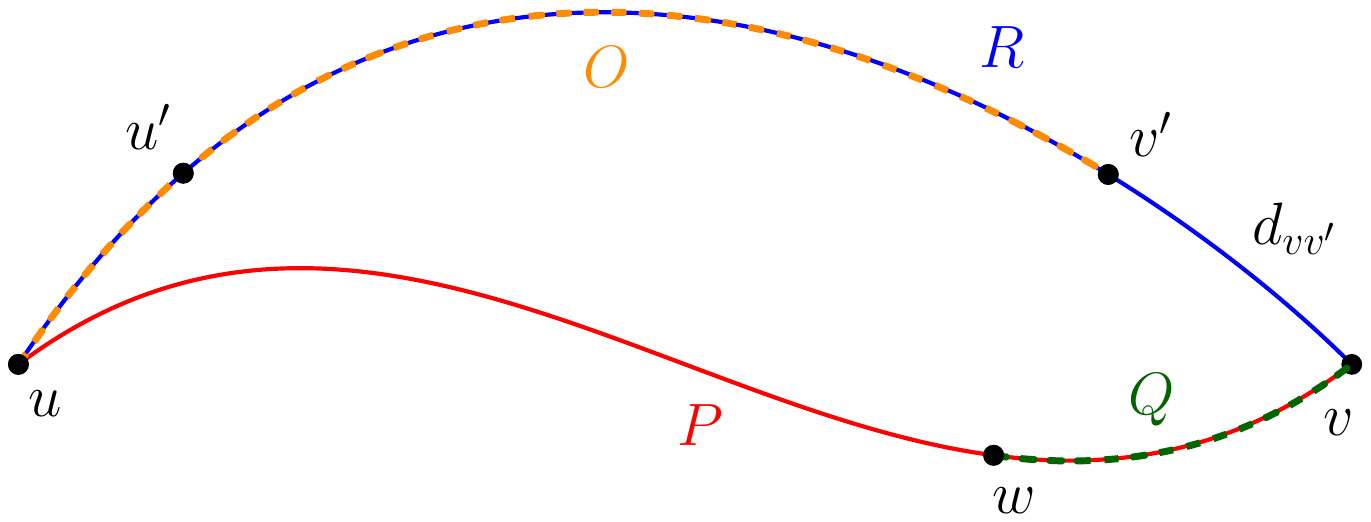}
	\caption{\boldmath Illustration of the given notations.}
	\label{fig:apsp_approx}
\end{figure}

\begin{lemma}
	\label{lem:computeR}
	Let $u,v \in V$ such that all shortest $u$-$v$-paths have more than $h$ hops and let path $R$ be defined as above. Then $w(R) = \tilde D^G_{uv}$ w.h.p.
\end{lemma}

\begin{proof}
	From our definition of the shortest $u$-$v'$-path $O$ we know that $hop(u,u') \leq h$ w.h.p. Since $u',v' \in M$ we have $d(u',v')= D^S_{u'v'}$ w.h.p.\ due to \Cref{lem:distancesSkeleton}. Therefore we find 
	\begin{align*}
	w(O) & = d(u,u') + d(u',v') = d_h(u,u') + D^S_{u'v'}\\ 
	& \geq \min_{u'\in M} \big(d_h(u,u') + D^S_{u'v'}\big)
	\end{align*}
	Since $O$ is a shortest $u$-$v'$-path, due to the triangle inequality and  \Cref{lem:distancesSkeleton}, we also have 
	\begin{align*}
	w(O) & = d(u,v') \leq \min_{u'\in M} \big(d_h(u,u') + d(u'v')\big)\\ 
	& = \min_{u'\in M} \big(d_h(u,u') + D^S_{u'v'}\big),
	\end{align*}
	and therefore we have equality. Finally we see from the definition of $R$
	$$w(R)  = w(O) + d(v,v') = \min_{u'\in M} \big(d_h(u,u') + D^S_{u'v'}\big) + d_{vv'}.\qedhere$$
\end{proof}

\begin{lemma}
	\label{lem:approxAPSP}
	Let $u,v \in V$ such that all shortest $u$-$v$-paths have more than $h$ hops and let $P,Q$ and $R$ be defined as above. Then we have $w(R) \leq w(P) + 2w(Q)$ w.h.p.
\end{lemma}

\begin{proof}
	Due to the triangle inequality and since $O$ is a shortest $u$-$v'$-path, we have $w(O) = d(u,v') \leq w(P)\!+\!d(v,v')$. Since $v' \in M$ is the marked node that minimizes $d_h(v,v')$ (c.f.\ \Cref{alg:transmitclosest}), we know that $d_h(v,v') \leq d_h(v,w)$. Furthermore, by definition of $P$, for the node $w \in M$ on $P$ closest to $v$ we have $hop(w,v) \leq h$ w.h.p., hence $d_h(v,w) = d(v,w)$. We put these pieces together and obtain
	\begin{align*}
	w(R) & = w(O) + d(v,v') \leq w(P)+2d(v,v') \\
	& \leq w(P)+2d_h(v,v') \leq w(P)+2d_h(v,w) \\
	& =  w(P)+2d(v,w) = w(P)+2w(Q). \tag*{\qedhere}
	\end{align*}
\end{proof}

We can now give the proof for the 3-approximate algorithm in the weighted case (\Cref{thm:approxgeneralAPSP}).

\begin{proof}[Proof of \Cref{thm:approxgeneralAPSP}]
	Let $u,v \in V$. We need to show $D^G_{uv} \leq \tilde D^G_{uv} \leq 3 \!\cdot\! D^G_{uv}$ (where $D^G$ denotes the true distance matrix). First consider the case that there is a shortest path $P$ between $u$ and $v$ with $|P| \leq h$. This case is easy, since then $d_m(u,v) = d(u,v)$ and therefore $\tilde D^G_{uv} = D^G_{uv}$ as can be seen from \Cref{eq:approxAPSP}. 
	
	Now consider the case that for the given pair $u,v \in V$ all shortest $u$-$v$-paths have more than $h$ hops. Since $R$ is a $u$-$v$-path (but not necessarily a shortest) we know from \Cref{lem:computeR} that $\tilde D^G_{uv} = w(R) \geq w(P) = D^G_{uv}$.
	We employ Lemmas \ref{lem:computeR} and \ref{lem:approxAPSP} to obtain 
	$$\tilde D^G_{uv} = w(R) \leq w(P) + 2w(Q) \leq 3 \!\cdot\! w(P) = 3 \!\cdot\! D^G_{uv}.$$	
	The total time complexity of \Cref{alg:approxAPSP} is $\tilO\big(\!\max(x,n/x)\big)+ \tilO\big(n/x\big) + \tilO\big(\!\sqrt n\big)$ due to Fact \ref{fct:constructskeletonprime}, \Cref{lem:distancesSkeleton} and Fact \ref{fct:transmitclosest}, which equals the claimed time complexity if we choose $x = \tilT(\!\sqrt n)$.
\end{proof}

\Cref{thm:approxunweightedAPSP} is an obvious corollary of the following theorem. The idea to prove the approximation ratio of $(1\!+\!\varepsilon)$, is to have nodes explore their neighborhood up to distance $n/h = \sqrt{2nW/\varepsilon}$ (by choosing $x$ appropriately). This guarantees that we only have to approximate distances between nodes with more than $\sqrt{2nW/\varepsilon}$ hops, which allows to make the approximation error arbitrarily small.

\begin{theorem}\label{thm:approxsmallweightedAPSP}
	For arbitrary $\varepsilon > 0$, there is an algorithm that computes a $(1\!+\!\varepsilon)$-approximation of the APSP problem in $\tilO\big(\!\sqrt{nW/\varepsilon} \,\big)$ rounds w.h.p., where $W := W_{\text{\emph{max}}}/W_{\text{\emph{min}}}$ is the ratio of maximum to minimum edge weight in $G$.
\end{theorem}

\begin{proof}
	Let $u,v \in V$. We show $D^G_{uv} \leq \tilde D^G_{uv} \leq (1\!+\!\varepsilon) D^G_{uv}$. As in the proof of \Cref{thm:approxgeneralAPSP} it is clear that $\tilde D^G_{uv} \geq  D^G_{uv}$. In case there is a shortest path $P$ between $u$ and $v$ with $|P| \leq m$ hops, we have that $\tilde D^G_{uv} = D^G_{uv}$ due to the first argument of the outer $\min$ function in \Cref{eq:approxAPSP}.
	
	What remains, is to prove the claim for a given pair $u,v \in V$ for which all shortest $u$-$v$-paths have more than $m$ hops. We choose $x := \sqrt{n\varepsilon/2W} \!\cdot\! \frac{1}{\xi \log n}$. With this choice we get $h = \xi x \ln n = \sqrt{n\varepsilon/2W}\leq \sqrt{n}$ and thus \smash{$m = \max\big(h,\frac{n}{h}\big) = \frac{n}{h}$}. Therefore \smash{$|P| \geq m = \frac{n}{h} = \sqrt{2nW/\varepsilon}$}. We obtain the following
	$$w(P) \geq W_{\text{min}}\cdot|P| \geq {W_{\text{min}}}\cdot \frac{n}{h} = \sqrt{{2nW_{\text{max}}W_{\text{min}}}/{\varepsilon}}.$$
	Then we use Lemmas \ref{lem:computeR} and \ref{lem:approxAPSP}, and the fact that $|Q|\leq h$ w.h.p., and we obtain 
	\begin{align*}
	\tilde D^G_{uv} & = w(R) \leq w(P) + 2w(Q) \\
	&  \leq w(P) + 2|Q|W_{\max} \leq w(P) + 2hW_{\max}\\
	&  = w(P) + {\varepsilon}\sqrt{{2nW_{\text{max}}W_{\text{min}}}/{\varepsilon}} \leq (1+\varepsilon)w(P). \tag*{\qedhere}
	\end{align*}
	
	As in the proof of \Cref{thm:approxgeneralAPSP}, the complexity of \Cref{alg:approxAPSP} is \smash{$\tilO\big(\!\max(x,n/x)\big)+ \tilO\big(n/x\big) + \tilO\big(\!\sqrt n\big)$}, which is dominated by \smash{$\tilO\big(n/x\big) = \tilO\big(\!\sqrt{nW/\varepsilon} \,\big)$} due to our choice of $x$.
\end{proof}

It remains to analyze the local capacity $\lambda$ for which \Cref{thm:approxgeneralAPSP} and \Cref{thm:approxsmallweightedAPSP} hold.

\begin{lemma}
	We can $3$-approximate APSP in $\tilO\big(\!\sqrt{n}\big)$ rounds and $(1\!+\!\varepsilon)$-approximate APSP in $\tilO\big(\!\sqrt{nW/\varepsilon} \,\big)$ rounds with local capacity $\lambda = \Theta(n)$ (for the latter: assuming $\varepsilon$ is not too small).
\end{lemma}

\begin{proof}
	We use local edges in the subroutine \Cref{alg:constructskeletonprime}: \texttt{Construct-Skeleton'($x$)}, where we learn $G$ up to distance $m$ in the same number of rounds. The information that each node actually requires from its $m$-hop neighborhood are the $m$-limited distances and the information which nodes in said neighborhood are marked. As before (c.f.\ proof of \Cref{lem:exactAPSPcongestion}) we can use the distributed version of Bellman-Ford to learn the required information in $m$ rounds. This requires only $\lambda = 2n$ (as we argue in \Cref{lem:congestion_bellman_ford}).
	
	
	The number of tokens disseminated in the token dissemination routine called in the sub-procedure \Cref{alg:transmitskeleton}: \texttt{Transmit-Skeleton} depends on the parameter $x$ (c.f. \Cref{lem:distancesSkeleton}). By \Cref{lem:congestion_token_dissem} we require local capacity $\lambda = \Theta({n/x})$. This is $\lambda = \tilT(\!\sqrt n)$ for the 3-approximation. For the $(1\!+\!\varepsilon)$-approximation we require $\lambda = \tilT\big(\!\sqrt{nW/\varepsilon}\big)$, which is at most $\Theta(n)$ assuming that $W$ is constant and $\varepsilon$ is not too small (e.g.\ $\varepsilon \geq \frac{1}{n}$). Note that without this assumption we require \smash{$\lambda = \tilT \big(\!\max\big(n,\sqrt{nW/\varepsilon}\big)\big)$}.
\end{proof}

\subsection{Lower Bounds for APSP}
\label{sec:apsp_lower_bounds}

In order to obtain rigorous lower bounds we introduce a technical lemma. It shows that for a class of graphs and a dedicated node $b$, we can create a bottleneck for the information that can be transmitted from parts of the graph to $b$ (c.f.\ \Cref{fig:lowerBounds}, left). Subsequently we show that obtaining solutions (or even approximations) for the all pair shortest path problem requires that a certain amount of information (measured in terms of its entropy) must be transmitted to $b$, which demonstrates the lower bounds claimed in this section.

\begin{lemma}
	\label{lem:lowerboundlemma}
	Let $G=(V,E)$ be an $n$-node graph that consists of a subgraph $G'=(V',E')$ and a path of length $L$ (edges) from some node $a\in V'$ to $b\in V\setminus V'$ and that except for node $a$ is vertex-disjoint from $V'$.
	Assume further that the nodes in $V'$ are collectively given the state of some random variable $X$ and that node $b$ needs to learn the state of $X$. Every randomized algorithm that solves this problem in the hybrid network model requires \smash{$\Omega\Big(\!\min\big(L,\frac{H(X)}{L \cdot \log^2 n}\big)\Big)$} rounds, where $H(X)$ denotes the Shannon entropy of $X$.
\end{lemma}

\begin{proof}[Proof of \Cref{lem:lowerboundlemma}]
	We show that the problem at hand solves the basic two-party communication problem, where Alice (which simulates $G'$) knows the value of the random variable $X$ and Bob (represented by node $b$) needs to learn it. The Shannon entropy $H(X)$ constitutes a lower bound on the expected number of bits that must be sent from Alice to Bob in order that Bob can learn the $X$. This implies that at least $H(X)$ bits must be transmitted in the worst case. Furthermore, we show that in less than $L$ rounds, Bob (node $b$) can learn only $O(L \log^2 n)$ bits per round from Alice. Therefore it takes \smash{$\Omega\big(\frac{H(X)}{L \cdot \log^2 n}\big)$} rounds until $b$ knows the state of $X$.
	
	Let $\mathcal A$ be an algorithm that solves the problem at hand and assume it takes less than $L$ rounds. We reduce the problem above to a different setting, in which we can simulate the original execution of $\mathcal A$. We show that the altered setting is either equivalent or makes it easier for $\mathcal A$ to solve the problem. We do this by giving node $a$ some global knowledge and by carefully altering the mode of communication, such that $b$ still obtains the same information in the simulation as it would in the original execution of $\mathcal A$.
	
	
	First, we assume that the node $a$ has complete knowledge of $X$ at the beginning of $\mathcal{A}$. Second, we assume $a$ has knowledge of the whole structure of $G$. Third, we assume that except for the randomness determining $X$, all the additional randomness used by all nodes in $V'$ during the execution of $\mathcal{A}$ is also known to $a$. So far, any of the assumptions made can only help to solve the problem and they thus make our lower bound stronger. Moreover, under these assumptions $a$ can locally simulate the whole execution of $\mathcal{A}$.
	
	Let $P$ be the set of \textit{nodes} in $V\setminus V'$ of the path connecting $a$ and $b$. That is, $P$ contains all nodes of the path except for node $a$, thus $|P|= L$. Since we assumed that $\mathcal{A}$ runs for less than $L$ rounds, $b$ can only be influenced by the content of global messages that are either sent directly to $b$ itself via a global edge, or to a node in $P$, which can subsequently forward it to $b$ within less than $L$ rounds over the local edges of the path. We can therefore entirely focus on the global communication where the receiving node is in $P$.
	
	In particular, we show that we can equivalently assume that all messages received by nodes in $P$ over global links can be sent directly from $a$ to $b$ instead, whereas all other communication via global edges is prohibited (assuming sufficiently increased budgets for sending and receiving messages for $a$ and $b$). Since $a$ can simulate $\mathcal A$ for all nodes, it can also send any global message that is received by any node on $P$ via global edge directly to $b$ instead. More specific, whenever $u \in V$ sends a message to some $v \in P$ over a global edge, node $a$ sends the same message to $b$ instead, together with the information that $u$ is the sender and $v$ the recipient. Sending global messages directly to $b$ instead of some other node on $P$ can only help $b$ in learning the state of $X$.
	
	What remains is to show that $b$ receives the same information via local edges as in the original execution of $\mathcal{A}$. Note, that the messages that are sent to $b$ over local edges by the nodes on $P$ might depend on the messages that they received over global links. However, here we can use the fact that the message size on local links is not bounded. If a node $v\in P$ sends a message $m$ over a local link in the original execution $\mathcal{A}$, it can instead send a collection of all possible values for $m$ depending on the global messages it could have received. Because node $b$ knows the content of all these global messages, it can reconstruct the local messages $b$ would have received in the original execution of $\mathcal{A}$.
	
	We have therefore reduced the problem to the following setting. The whole execution of $\mathcal{A}$ for the nodes of $V'$ can be simulated by node $a$ alone. The reduced setting has only the path of length $L$ connecting $a$ and $b$ with local edges. In addition, $a$ is connected to $b$ over a single global edge and no other global edges exist. For the simulated setting, we can restrict the number of messages via this global edge to $O(L\log^2 n)$ messages per round, since that is the maximum amount of messages that can collectively be received by all nodes in $P$ during the execution of $\mathcal{A}$. In less than $L$ rounds, the only information that can reach $b$, is via this global edge. Since $b$ needs to learn $H(X)$ bits, it takes at least $\Omega(H(X)/(L\log^2 n))$ rounds for $b$ to learn $X$.
\end{proof}

\Cref{thm:lowerBoundAPSP} exploits the fact that for a path of length $n$, a node $b$ that located at one end, must learn the permutation of all nodes on the path in order to solve APSP exactly (c.f.\ \Cref{fig:lowerboundAPSP}). We show that this requires that $b$ learns $\Omega(n\log n)$ bits, yielding the claimed lower bound by virtue of \Cref{lem:lowerboundlemma}.

\begin{theorem}
	\label{thm:lowerBoundAPSP}	
	An algorithm that solves APSP in the hybrid network model takes \smash{${\Omega}\big(\!\sqrt{n/\log n}\big)$} rounds, even on unweighted graphs.
\end{theorem}

\begin{figure}[H]
	\vspace*{-0.3cm}
	\centering
	\includegraphics[width=0.62\linewidth]{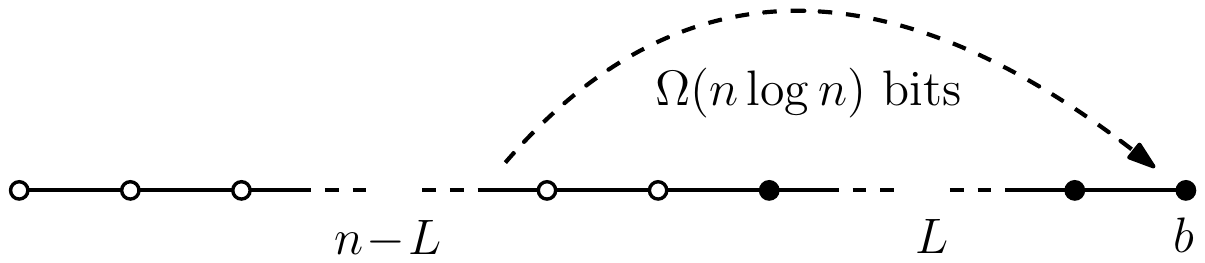}
	\caption{To learn the permutation of nodes farther than $L$ hops, $b$ must learn $\Omega(n \log n)$ bits.}
	\label{fig:lowerboundAPSP}
\end{figure}

\begin{proof}
	Let $G$ be an unweighted path with $n$ nodes and let $b$ be the last node on the path (c.f.\ \Cref{fig:lowerboundAPSP}). Let $L \leq n/2$ and let $\mathcal{A}$ be an algorithm that solves APSP on $G$. We allow that all nodes including $b$ know that $G$ is a path, but all nodes with hop distance at least $L$ from $b$ are permuted according to some random distribution and their permutation is unknown to $b$. 
	Let $S_{n-L}$ be the set of permutations of those nodes. We define the random permutation $X \in S_{n-L}$ and decree that each permutation $\pi \in S_{n-L}$ has the same probability $p := \mathbb{P}(X =\pi) := 1/|S_{n_L}|$. The Shannon entropy of $X$ is given by
	$$H(X) = \hspace*{-0.3cm}\sum_{\pi \in S_{n-L}} \hspace*{-0.2cm} p \log(1/p) = {|S_{n_L}|}\cdot {p} \log(1/p) = \log(|S_{n_L}|) \geq \log\big[\big(\tfrac{n}{2}\big)!\big] \geq \log\big[\big(\tfrac{n}{4}\big)^{n/4}\big] \in \Omega(n \log n)$$
	The random variable $X$ is collectively known by the nodes at hop distance at least $L$ from $b$. From \Cref{lem:lowerboundlemma} we learn that it takes at least \smash{$\Omega\big(\!\min(L,\frac{n}{L \cdot \log n})\big)$} rounds until $b$ knows $X$, which is necessary in order that $\mathcal{A}$ solves APSP on $G$. The claimed lower bound is obtained if we choose $L \in \Theta(\!\sqrt{n/\log n})$.
\end{proof}

\begin{proof}[Proof of \Cref{thm:lowerBoundAPSPapprox}]
	We construct an unweighted graph $G$ in which an $\alpha$-approximative APSP algorithm has the claimed lower bound (c.f.\ \Cref{fig:lowerBoundAPSPapprox}). One part of $G$ is a path with $x$ nodes, where \smash{$x \coloneqq \lfloor n/2 + L\rfloor$} and \smash{$L := \big\lfloor \!\frac{\sqrt{n}}{\sqrt{c}\log n} \big\rfloor$}. Node $b$ is at an end of the path. Moreover, $G$ has two sets of nodes $S_1, S_2$, of size \smash{$y \coloneqq \big\lfloor \frac{n-x}{2}\big\rfloor$}. 
	
	We have $x,y \in \Omega(n)$. Note that we round $x,y$ such that $G$ has $x\!+\!2y \in [(n\!-\!3) .. n]$ nodes in total. This is w.l.o.g.\ since we can always attach a few additional nodes to $b$. Every node in $S_1$ has an edge to $v_1$ which is the node with \smash{$hop(v_1,b) = L$}. Every node in $S_2$ has an edge to $v_2$, which is the node with $hop(v_2,b) = x$. 	
	
	We allow that the layout of $G$ and the nodes that are on the path from $v_2$ to $b$ are fixed and globally known. However, we assign the $2y$ remaining nodes randomly to $S_1$ and $S_2$. Formally, we fix a set of $y$ nodes $u_1, \ldots, u_y$ and assign each randomly to $S_1$ or $S_2$ with probability $1/2$. The last $y$ nodes are used to fill up $S_1$ and $S_2$ to size $y$.
	
	Let $\mathcal{A}$ be an algorithm that computes an $\alpha$-approximation of APSP for $\alpha \leq {\sqrt{nc}\cdot\log n}/{2}$. In order to approximate APSP, node $b$ needs to determine a distance estimation $\tilde d(b,u)$ for each $u \in \{u_1, \ldots, u_y\}$ such that $d(b,u) \leq \tilde d(b,u) \leq \alpha \!\cdot\! d(b,u)$. If $b$ does not know whether $u_i \in S_1$ or $u_i \in S_2$ for one node $u_i$ with $i \in [1..y]$, then the best, valid estimation $b$ can make is $\tilde d(b,u_i) = x$ under the assumption that $u_i \in S_2$. If however $u_i \in S_1$ is true, then the approximation ratio $\alpha$ would be
	$$\alpha = \frac{\tilde d(b,u_i)}{d(b,u_i)} = \frac{x}{L} = \frac{\lfloor n/2 + L\rfloor}{L} > \frac{n/2}{L} > {\sqrt{nc}\cdot\log n}/{2}.$$
	Hence $b$ must learn whether $u_i \in S_1$ or $u_i \in S_2$ for all $i \in [1..y]$. Let $X \in \{1,2\}^y$ be the random assignment $u_1, \ldots, u_y$ either to $S_1$ or $S_2$. I.e., $X$ represents the outcome of a $y$-fold Bernoulli process. Since each outcome $o \in \{1,2\}^y$ is equally probable, we have $p := \mathbb{P}(X \!=\! o) = 1/|X| = 1/2^n$. Thus the entropy of $X$ is 
	$$H(X) = \hspace*{-0.3cm}\sum_{o \in \{1,2\}^y} \hspace*{-0.2cm} p \log(1/p) = {2^y}\frac{1}{2^y} \log(2^y) = y \in \Omega(n)$$
	Now the conditions of \Cref{lem:lowerboundlemma} apply, hence it takes at least \smash{$\Omega\big(\!\min(L,\frac{n}{L \cdot \log^2 n})\big) = \Omega\big(\!\sqrt{n}/\log n\big)$} rounds until $b$ knows the state of $X$. We showed that the latter is a requirement in order that $\mathcal{A}$ can compute an $\alpha$-approximation for APSP on $G$.
\end{proof}

 \section{Single-Source Shortest Paths} \label{sec:sssp}
 The final section revolves around computing single-source shortest path distances. We first present an exact algorithm to solve SSSP in time $\tilde{O}(\sqrt{\SPD})$.
 Subsequently, we give two algorithms that approximate SSSP for various running times and approximation factors.
 
 \subsection{Exact SSSP}
\label{sec:sssp_exact}

In the following section we give \Cref{alg:exactSSSP} as well as a detailed description of its functionality. Subsequently we proof important properties of \Cref{alg:exactSSSP}, from which we can infer \Cref{thm:exactSSSP}.

\begin{algorithm}
	\caption{\texttt{Exact-SSSP}  \Comment{executed by a node $v$} }	
	\label{alg:exactSSSP}
	\begin{algorithmic}
		\State $i \gets 1$ 
		\While{some distance value changed in previous phase \textbf{or} $i=1$} \Comment{\textit{phase} $i$}
		\For{two rounds} \Comment{\textit{$v$ learns $G(v,2i)$}}
		\State $v$ sends all recently learned edges (\& weights) to neighbors \Comment{\textit{initially, $v$ sends incident edges}}
		\EndFor
		
		\State $R \gets \{\langle v, d_{t(i-1)}(s,v), \emptyset\rangle\}$ \Comment{\textit{initial recursion message}}
		
		\smallskip
		
		\For{$\lceil \log n \rceil \!+\! 1$ steps} \Comment{\textit{divide and conquer on subtrees}}
		\State $R' \gets \emptyset$ \Comment{\textit{recursion messages of $v$ for the next recursive step}}
		\For{every message $\langle u, d, L \rangle \in R$ \textbf{in parallel}}
		\State store $d$ as candidate value \Comment{\textit{break ties by ID(u)}}
		\State $S$ is subtree of $T(u, i)$ rooted at $v$ without the subtrees rooted at any node of $L$
		\If{$|V[S]| > 1$} \Comment{\textit{$S$ has more than one vertex}}
		\State $x \gets$ splitting node of $S$ \Comment{\textit{$x$ locally computable with knowledge of $G(v,2i)$}}
		\State add $\langle u, d, L \cup \{x\} \rangle$ to $R'$
		\State send $\langle u, d + d_i(v, x)\rangle$ by participating in an aggregation towards $x$ \Comment{\textit{use methods of \cite{AGG+18}}}
		\EndIf
		\EndFor
		
		\If{$v$ received message $\langle w,d \rangle$ through aggregation} \Comment{\textit{$v$ initiates recursion at children}}
		\State store $d$ as candidate value \Comment{\textit{break ties by ID(u)}}
		\State send recursion message $\langle w, d + w(v, c), \emptyset \rangle$ to every child $c$ of $v$ in $T(w, i)$ using local edges
		\EndIf
		
		\smallskip
		
		\State add received recursion message with smallest associated distance value to $R'$
		\State $R \gets R'$
		\EndFor
		
		\smallskip
		
		\State $d_{t(i)}(s,v) \gets$ minimum of all candidate values; $i \gets i+1$
		\EndWhile
	\end{algorithmic}			
\end{algorithm}



In phase $i$, the goal of each node $v \in V$ is to inform every node $u \in V$ within hop-distance $i$ about the weight of a path of length at most $t(i)$ from $s$ to $u$ that contains $v$. If there is a path of at most that length with a smaller weight that does not contain $v$, then it contains a different node that would instead succeed in informing $u$. 
Our idea is to use a divide-and-conquer approach. We define $T(v,i)$ as the shortest-path tree of $G(v,i)$ and we need to inform all nodes of $T(v,i)$ about $d_{t(i-1)}(s,v)$. 
The parent of each node $u$ in $T(v,i)$ is its immediate predecessors on a shortest path from $v$ to $u$ in $G(v,i)$; if there are multiple such nodes, we choose the one with smallest identifier.


We divide each phase into $\lceil \log n \rceil + 1$ steps, each of which corresponds to one recursive call in $T(v,i)$.
At the beginning of each step of phase $i$, $v$ stores a set $R$ of \emph{recursion messages}.
A recursion message is of the form $(u, d, L)$, where $u$ is the node from which the recursion originated, $d$ denotes the weight of a shortest path from $s$ to $v$ that contains $u$, and $L$ is a set of nodes of $T(u, i)$ whose subtrees can be disregarded by $v$ as they are taken care of by a different node.

At the beginning of the first step of phase $i$, the node $v$ only stores a single recursion message \\
$(v, d_{t(i-1)}(s,v), \emptyset)$, which initiates informing all nodes in $T(v,i)$.
Then, in every step $v$ does the following for each message $(u, d, L) \in R$ in parallel:
First, it stores $d$ as a \emph{candidate value} for $d_{t(i)}(s,u)$.
At the end of the phase, $v$ will determine the correct value by choosing the minimum of all received candidate values.
Let $S$ be the subtree of $T(u, i)$ rooted at $v$ that does not contain the subtrees rooted at any node of $L$.
Note that our algorithm ensures that $v$ itself is a node of $T(u, i)$, and, as $v$ knows its complete neighborhood up to a hop-distance of $2i$, it also knows $T(u, i)$.

If $|V[S]| > 1$, $v$ has to continue the recursion in $S$ by choosing a \emph{splitting node} $x$ of $S$, which is a node whose removal disconnects $S$ into trees of size at most $\frac{1}{2} |V[S]|$.
As we later show, such a node can easily be computed locally at $v$.
We continue the recursion (1) in $S$ without the subtree of $x$, and (2) in each subtree of $S$ that is rooted at a child of $x$.
For (1), $v$ simply sends a recursion message $(u, d, L \cup \{x\})$ to itself.
For (2), $v$ does the following:
It first sends a message $(u, d + d_i(v, x))$ to $x$.
More precisely, $u$ does not send the message directly, but participates in an aggregation in a simulated butterfly network as described in \cite{AGG+18}.
In the aggregation procedure $x$ does not receive all messages, but only the one that contains the minimum distance value.
We break ties by preferring the message that contains the node with minimum identifier.
Assume that $v$ receives an aggregation message $(w, d)$ from some node $w$ (where $d$ is minimal among all messages taking part in the aggregation towards $v$). Then $v$ stores $d$ as a candidate value and sends a recursion message $(w, d + w(v, c), \emptyset)$ to every child $c$ of $v$ in $T(w, i)$, which is again known to $v$ as $v$ lies in $T(w,i)$.
From all recursion messages a node receives in that way, it only keeps the one with minimal associated distance value, again breaking ties by choosing the message that contains the node with minimum identifier, and stores the distance as a candidate value.


As we do not require the nodes to know $\SPD$, we have to let the nodes detect when to terminate.
We simply stop the algorithm when for the first time no distance value changes at any node, i.e., after the first phase $i$ such that $d_{t(i)}(s,v) = d_{t(i - 1)}(s,v)$ for all nodes $v \in V$.
This can be detected by simply performing a convergecast in the butterfly as described in \cite{AGG+18} at the end of every phase.
When for the first time every node declares that its value did not change, all nodes are instructed to terminate.

We begin our analysis by showing the correctness of our algorithm.
To that end, we first show that the subtrees that are covered by a recursion message essentially halve in size every step, which implies that after $\lceil \log n \rceil + 1$ steps no node stores a recursion message anymore.

\begin{lemma}\label{lem:splittingSizes}
	Let $S$ be of size at least $|V[S]| \ge 2$.
	$v$ can compute a splitting node $x$, whose removal disconnects $S$ into trees each of size at most $\frac{1}{2} |V[S]|$.
\end{lemma}

\begin{proof}
	For an inner node $u$ of $S$ define $s(u)$ as the number of nodes in the subtree of $S$ rooted at $u$, and let $p(u) = |V[S]| - s(u)$.
	The splitting node is computed by performing a search that descends into $S$, starting at its root $v$ (which, as $|V[S]| \ge 2$, must have at least one child).
	If the search is currently at some inner node $u$, then let $w$ be the child of $u$ that maximizes $s(w)$ (choosing the node with minimum identifier in case of a tie).
	If $p(w) < \frac{1}{2} |V[S]|$, then the search continues at $w$; otherwise, $u$ is chosen as the splitting node.
	Note that if $p(w) < \frac{1}{2} |V[S]|$ then $w$ cannot be a leaf node.
	Clearly, the search can be performed locally at $v$ and will eventually terminate at a splitting node $x$.
	
	Let $y$ be the child of $x$ in $S$ that maximizes $s(y)$.
	As $x$ is chosen as a splitting node, $p(y) \ge \frac{1}{2} |V[S]|$.
	Therefore, $s(y) = |V[S]| - p(y) \le |V[S]| - \frac{1}{2} |V[S]| = \frac{1}{2} |V[S]|$, and, as $y$ is the child that maximizes $s(y)$, the same holds for all other children of $x$.
	
	If $x$ does not have a parent, then the claim holds immediately.
	Otherwise, its parent must have been considered as a splitting node as well.
	However, as it has not been chosen, $p(x) < \frac{1}{2} |V[S]|$, which concludes the proof.
\end{proof}

\begin{lemma}\label{lem:correctnessPhase}
	Let $v \in V$.
	$v$ learns $d_{t(i)}(s,v)$ in phase $i$.
\end{lemma}

\begin{proof}
	We prove by induction on $i$.
	In the first phase, $s$ will choose itself as a splitting node, and send a recursion message $(s, w(s,u), \emptyset)$ to every neighbor $u \in V$, and one recursion message $(s, 0, \{s\})$ to itself.
	No subsequent recursion message will be sent, and therefore every neighbor $u$ of $s$ learns $d_{t(1)}(s, u) = d_1(s, u) = w(s,u)$.
	For every other node $w$ we have $d_1(s,w) = \infty$.
	
	Now consider phase $i>1$ and let $v \in V$.
	There must be a node $u$ such that $d_{t(i - 1)}(s, u) + d_i(u, v) = d_{t(i)}(s, v)$; if there are multiple such nodes, let $u$ be the one with minimum identifier.
	$v$ lies in $T(u,i)$, and the branch from $u$ to $v$ must be part of a shortest path of at most $t(i)$ hops from $s$ to $v$.
	Note that every distance value that is received by any node $w$ in phase $i$ corresponds to an actual path of that length and with at most $t(i)$ hops from $s$ to $w$; therefore, no node of the branch from $u$ to $v$ will ever receive a recursion message with a smaller distance value, as in this case there would exist an even shorter path from $s$ to $v$.
	Furthermore, by our choice of $u$, every node of the branch will always prefer recursion messages corresponding to root $u$ over recursion messages with the same distance value corresponding to any other root.
	Therefore, and by \Cref{lem:splittingSizes}, $v$ will receive $d_{t(i)}(s, v)$ within the $\lceil \log n \rceil + 1$ steps, and will not receive any smaller distance value.
\end{proof}	
	%
	

We now show that termination of the algorithm is correct.

\begin{lemma}\label{lem:termination}
	No distance value changes in phase $i$ if and only if $d_{t(i)}(s,v) = d(s,v)$ for every node $v \in V$.
\end{lemma}

\begin{proof}
	First, if $d_{t(i)}(s,v) = d(s,v)$ for every node $v$, then clearly no node will ever receive a smaller distance value anymore, as in that case there would exist an even shorter path from $s$ to to that node.
	For the other direction, assume that no distance value changes in some phase $i$.
	Let $\ell_v$ be the number of hops of a shortest path from $s$ to $v$.
	We prove that $v$ knows $d(s,v)$ at the beginning of phase $i$ by induction on $\ell_v$.
	If $\ell_v = 0$, then $v = s$, and as all edge weights are positive, $v$ knows $d(s,v)$ already at the beginning of the first phase, and thus also at the beginning of phase $i$.
	Now let $\ell_v = j$.
	Then there exists a neighbor $u$ of $v$ such that $\ell_u = j-1$ and $d(s,v) = d(s,u) + w(u,v)$.
	By the induction hypothesis, $u$ knows $d(s,u)$ at the beginning of phase $i$, and by definition of our algorithm, $v$ must receive $d(s,v)$ as a candidate value throughout the execution of phase $i$.
	However, as $d_{t(i-1)}(s,v) = d_{t(i)}(s,v)$, we must have $d_{t(i-1)}(s,v) = d(s,v)$, and thus $v$ knew $d(s,v)$ already at the beginning of phase $i$.
\end{proof}

Let us turn to the runtime of the algorithm.

\begin{lemma}
	The algorithm terminates after $2\sqrt{\SPD} + 1$ phases.
	Every phase takes time $O(\log^2 n)$, w.h.p.
\end{lemma}

\begin{proof}
	As any shortest path has length at most $\SPD$, after phase $2\sqrt{\SPD}$ every node $u$ knows $d_{t(i)}(s,u) = d(s,u)$ by \Cref{lem:correctnessPhase}.
	Therefore, \Cref{lem:termination} implies that no distance value changes in the subsequent round, in which case the algorithm terminates.
	A phase consists of $O(\log n)$ steps, where in each step a node may be member of as many aggregations as it has recursion messages stored.
	Note that if a node stores $k$ recursion messages at the beginning of a step, then it may send at most $k$ recursion messages to itself and receive at most one recursion message corresponding to a different root in this step.
	Together with \Cref{lem:splittingSizes}, this observation implies that in each step a node may only be member of at most $\lceil \log n \rceil + 1$ aggregations, and target of at most one aggregation.
	From the discussion of \cite{AGG+18} it follows that all aggregations of a single step can be performed in time $O(\log n)$, w.h.p.
	The convergecast to detect termination at the end of a phase takes an additional $O(\log n)$ steps.
\end{proof}

From the above, we conclude \Cref{thm:exactSSSP}. The algorithm can easily be modified to solve $(h,k)$-SSP for given $h$ and $k$.
As the shortest-path diameter of $G$ is generally not known, our algorithm has to sequentially increase the distance at which the nodes learns their respective neighborhood, and propagate their distance information. 
For a given hop-distance $h$, we can first let every node $v \in V$ learn $G(v,2\lceil \sqrt{kh} \rceil)$ by sending all information about $G$ known so far via its local edges for $2\lceil \sqrt{kh} \rceil$ rounds.
Then, we separately perform $\lceil h/ \sqrt{kh} \rceil$ phases of the algorithm for each $k$, where in each phase every node $v$ always propagates distance information by $\lceil \sqrt{kh} \rceil$ hops. 
The total runtime amounts to $\widetilde{O}(\sqrt{kh} + k \cdot h/ \sqrt{kh}) = \widetilde{O}(\sqrt{kh})$, which concludes the following theorem.

\begin{theorem}
	The modified algorithm solves $(h,k)$-SSP in time $\tilO\big(\!\sqrt{kh}\big)$, w.h.p.
\end{theorem}

It remains to discuss the local capacity $\lambda$ that \Cref{alg:exactSSSP} requires.

\begin{lemma}
	\label{lem:exactAPSPcongestion}
	\Cref{alg:exactSSSP} works in the same time $\tilO\big(\!\sqrt{\SPD}\big)$ for local capacity $\lambda = \tilT\big(n^2/\sqrt{\SPD}\big)$.
\end{lemma}

\begin{proof}
	\Cref{alg:exactSSSP} uses local edges for two purposes: (I) so that each node learns $G$ up to distance $\tilO(\!\sqrt{\SPD})$ in the same number of rounds and (II) for splitting nodes to notify their children in their respective subtrees for which they need to start a new recursion. In \Cref{apx:balance_congestion} we give \Cref{lem:congestion_to_learn_G}, which shows that (I) can be done with $\lambda = \tilT\big(n^2/\sqrt{\SPD}\big)$ in the time of the algorithm. For (II) we analyze the congestion on local edges caused by the notifications from splitting nodes to their children in one recursion step.
	
	Let $N$ be the overall number of splitting nodes in some fixed step of \Cref{alg:exactSSSP}. Note that some node $v \in V$ can be splitting node in multiple trees. However, since the number of recursion instances each node handles in parallel is $\tilO(1)$ and since every node appoints at most one splitting node per recursion instance, we have $N = \tilO(n)$. 	
	Every splitting node sends at most one message over each incident edge. This means that in every step we have congestion at most $N = \tilO(n)$ on each edge caused by splitting nodes.
\end{proof}

 \subsection{\boldmath Approximate SSSP in ${\tilO(n^{1/3})}$}
\label{sec:sssp_approx}

In the following we discuss the concept and properties of \Cref{alg:approxSSSP}. As some of the arguments are similar to \Cref{sec:apsp}, we will restrict ourselves to a briefer discussion. Subsequently we prove \Cref{thm:approximateSSSP}. The base concept is to compute a skeleton graph $S=(M,E_S)$ (which the source $s$ is always part of) and then use token dissemination to simulate the broadcast congested clique (BCC) model (c.f.\ \Cref{def:BCCmodel}) on $M$. This allows use the algorithm of \cite{BKK+17} for said model, to $(1\!+\!\varepsilon)$-approximate SSSP on $S$. Then we broadcast the distance estimations of all pairs $M\times\{s\}$ with token dissemination, which we can use to approximate SSSP-distances on the whole graph.

\begin{definition}[Broadcast Congested Clique Model]
	\label{def:BCCmodel}
	Given a set of nodes with the standard synchronous message passing model,  with a different notion of how messages are sent and received. That is, every round each node can send one message of size $O(\log n)$, which is known by all nodes in the subsequent round.
\end{definition}

\begin{algorithm}[H]	
	\caption{\texttt{Approximative-SSSP($\varepsilon$)}\Comment{$h  = \tilO(x)$, \textit{$1\!+\!\varepsilon$ is the approximation factor}}} 
	\label{alg:approxSSSP}
	\begin{algorithmic}
		\If {$v$ is the source} mark $v$ \Comment{\textit{source is always part of $M$}}
		\Else \,mark $v$ with probability \smash{$\frac{1}{x}$}
		\EndIf
		\For{$h$ rounds}\Comment{\textit{learn $h$-hop limited distances}}
		\State $v$ sends knowledge about $G$ it learned last round via local edges \Comment{\textit{initially $v$ sends its incident edges}}
		\EndFor
		\State \texttt{SSSP-By-Simulating-BCC($\varepsilon$)} \Comment{\textit{SSSP approximation of \cite{BKK+17} on marked nodes simulating BCC}}
		\State approximate SSSP distances locally with \Cref{eq:approxSSSP}
	\end{algorithmic}			
\end{algorithm}

\begin{equation*}
\tilde d_{sv} := \min\Big(d_h(s,v), \min_{u\in M} \big(\tilde d_{su} + d_h(u,v)\big)\Big) \tag{\ref{eq:approxSSSP}}
\end{equation*}

\begin{fact}
	\label{fct:approxSSSPskeleton}	
	As we have seen before, the size of $M$ is \smash{$\tilO(n/x)$} w.h.p. Moreover \Cref{alg:approxSSSP} establishes a weighted graph $S = (M,E_S)$ with $E_S \coloneqq \{ \{u,v\} \!\mid\! u,v\!\in\! M, \text{hop}(u,v) \!\leq\! h\}$ and weights $d_h(u,v)$ for $\{u,v\} \in E_S$.
	All nodes $v \in V$ know all $u \in V$ that are within $h$ hops as well as the distances $d_h(u,v)$.
\end{fact}

\begin{algorithm}[H]	
	\caption{\texttt{SSSP-By-Simulating-BCC($\varepsilon$)}\Comment{\textit{$1\!+\!\varepsilon$ is the approximation factor}}} 
	\label{alg:simulateBCC}
	\begin{algorithmic}
		\If {$v$ is marked} \Comment{\textit{marked nodes participate in SSSP approximation}}
		\State take part in BCC simulation \Comment{\textit{BCC-round equals a run of \texttt{Token-Dissemination}}}
		\State take part in the algorithm of \cite{BKK+17} to ($1\!+\!\varepsilon$)-approximate SSSP on simulated BCC
		\State create a token $\langle I\!D(v), \tilde d_{sv}\rangle$ \Comment{\textit{$\tilde d_{sv}$ is an $(1\!+\!\varepsilon)$-approximation of $d(s,v)$}}
		\EndIf
		\State \texttt{Token-Dissemination} \Comment{\textit{disseminate approximate distances}}
	\end{algorithmic}			
\end{algorithm}

\begin{lemma}
	\label{lem:simulateBCC}
	\Cref{alg:simulateBCC} computes and disseminates a $(1\!+\!\varepsilon)$-approximation $\tilde d_{su}$ of $d(s,u)$ between every pair $(u,s) \in M \times \{s\}$ in $\tilO(\!\sqrt{n/x} \cdot \varepsilon^{-9})$ rounds.
\end{lemma}

\begin{proof}
	We can simulate one round in the BCC model on the skeleton $S$ in our hybrid network model as follows. Every node in $M$ creates a token of size $O(\log n)$ from the message it wants to send and then participates in the \texttt{Token-Dissemination} protocol from \Cref{sec:tokenDissemination} to disseminate that message to all nodes. This takes $\tilO(\!\sqrt{|M|})$ rounds due to \Cref{thm:tokenDissemination}. Then we invoke the algorithm of \cite{BKK+17} on the simulated BCC model. This gives us approximations $\tilde{d}_{su}$ with $d(s,u) \leq \tilde{d}_{su} \leq (1\!+\!\varepsilon)d(s,u)$ for all $u \in M$ (c.f.\ Theorem 8 of \cite{BKK+17}).
	
	Their algorithm has a running time of $\tilO(\varepsilon^{-9})$ rounds in the BCC model, thus the total running time is \smash{$\tilO\big(\!\sqrt{|M|} \cdot \varepsilon^{-9}\big)$} rounds. The subsequent dissemination of the approximated distances $\tilde{d}_{su}$ for all $u \in M$ takes only \smash{$\tilO\big(\!\sqrt{|M|}\big)$} rounds. Since we sample $M$ with probability $1/x$ out of $n$ nodes (and add at most one additional node, namely the source) we have \smash{$|M| \in  \tilO(n/x)$} and the claim follows.
\end{proof}

\begin{proof}[Proof of \Cref{thm:approximateSSSP}]
	Let $v\in V$. We make a case distinction. First assume the simple case that a shortest $s$-$v$ path $P$ with $|P| \leq h$ hops exists. Then $d(s,v) = d_h(s,v)$ and \Cref{alg:approxSSSP} outputs the exact distance, as can be seen from \Cref{eq:approxSSSP}. Now assume the contrary, i.e., all shortest $s$-$v$ paths have more than $h$ hops. Then there is a shortest $s$-$v$ path $P$, that has a marked node $u \in M$ on it with $hop(u,v) \leq h$. This is due to \Cref{lem:longPathsM} for appropriately chosen $h \in \tilO(x)$. Therefore $d_h(u,v) = d(u,v)$. Node $v$ already knows an approximation $\tilde{d}_{su}$ for the sub-path of $P$ from $s$ to $u$ with $d(s,u) \leq \tilde{d}_{su} \leq (1\!+\!\varepsilon)d(s,u)$ as we showed in \Cref{lem:simulateBCC}. We have
	$$w(P) = d(s,u) + d(u,v) \leq \tilde{d}_{su} + d_h(u,v) \leq (1\!+\!\varepsilon)d(s,u) + d(u,v) \leq (1\!+\!\varepsilon)w(P),$$
	where $\tilde{d}_{su} + d_h(u,v)$ is computed by \Cref{alg:approxSSSP} in \Cref{eq:approxSSSP}. The round-complexity of \Cref{alg:approxSSSP} is $\tilO(x)$ for the local search and $\tilO(\!\sqrt{n/x} \cdot \varepsilon^{-9})$ for the computation and dissemination of the approximated distances $\tilde{d}_{su}$ with $u \in M$. This is optimized for $x = n^{1/3}\cdot\varepsilon^{-6}$.
\end{proof}

Finally we give an analysis of the local capacity $\lambda$ the above algorithm requires.

\begin{lemma}
	\label{lem:approxSSSPcongestion}
	\Cref{alg:approxSSSP} works in the same time $\tilO(n^{1/3}\varepsilon^{-6})$ for local capacity $\lambda = \tilT(n^{2/3}\varepsilon^6)$.
\end{lemma}

\begin{proof}
	In \Cref{alg:approxSSSP} we use local edges to do the local exploration in a $\tilO(x)$-hop neighborhood of each node in the same number of rounds. However, the only type of information that the nodes require from this local search in order to first establish the skeleton and finally compute \Cref{eq:approxSSSP}, is that each node needs to learn its $h$-limited distance to the set of marked nodes $M$ (which includes the source $s$).
	
	We can alternatively do this in the same time by first making the set of marked $M$ nodes public knowledge via a run of token dissemination. Since \smash{$|M|=n/x=n^{2/3}\varepsilon^6$}, this takes only \smash{$\tilO(n^{1/3}\varepsilon^3)$} (\Cref{thm:tokenDissemination}) rounds and local capacity $\lambda = \tilT(n^{1/3}\varepsilon^3)$ (\Cref{lem:congestion_token_dissem}). Then we conduct a distributed Bellman-Ford with sources $M$ for $h$ rounds, so each node learns its $h$-limited distance to each marked node. In \Cref{lem:congestion_bellman_ford} we show that $\lambda = 2|M| = \tilT(n^{2/3}\varepsilon^6)$ suffices for this.	
	
	For each round of BCC simulation and finally to make the distance estimations of the skeleton public knowledge we conduct an additional run of token dissemination with $|M|$ tokens in \smash{$\tilO(n^{1/3}\varepsilon^3)$}. We already established that this takes local capacity at most $\lambda = \tilT(n^{1/3}\varepsilon^3)$.
%
%
\end{proof}

 \subsection{\boldmath Approximate SSSP in $\tilO(n^\varepsilon)$}
\label{sec:subpolynomial}

In this section we present a fast algorithm that runs in $O(\alpha^5)$ for some parameter $\alpha \ge 5$, albeit with a coarser approximation ratio of $(\log_\alpha n)^{O(\log_\alpha n)}$. Nevertheless, notice that we get a constant approximation ratio when we set $\alpha = n^\epsilon$ for fixed $\epsilon > 0$.  The key ingredient for our algorithm is a sparse spanner of the skeleton graph that we call a {\em skeleton spanner} that we formally define shortly.
In Section~\ref{sub:spanskel}, we present an algorithm to construct such a sparse  skeleton spanner. Subsequently, in  Section~\ref{sub:recurse}, we describe our algorithm for computing the $(\log_\alpha n)^{O(\log_\alpha n)}$-SSSP.

\subsubsection{Constructing a Sparse Spanner of the Skeleton Graph} \label{sub:spanskel}

In the following, we give a simple algorithm to compute a sparse
spanner of a skeleton graph, which we will call a \textit{skeleton spanner}. A formal definition is given below in \Cref{def:skeletonspanner}. 
The spanner algorithm is solely based on
computing limited-depth BFS trees, which can be done efficiently even when relying entirely on global edges by using the methods of \cite{AGG+18} as black-box.

\begin{definition}[Skeleton Spanner]\label{def:skeletonspanner}
	Let $G=(V,E,w)$ be a weighted graph, let $M\subseteq V$ be a set of
	marked nodes of $G$, and let $h\geq 1$ be an integer parameter. A
	$h$-hop skeleton spanner $H=(M,E_H)$ with stretch $s\geq 1$ is a weighted
	graph with the following properties that (1) every edge $\set{u,v}\in E_H$ corresponds to a path $P$ in $G$
	between $u$ and $v$ and the weight of $\set{u,v}$ is the total
	weight $w(P)$ of $P$, and (2) for every two nodes $u, v\in M$, we have $d_{h,G}(u,v)\leq
	d_H(u,v)\leq s\cdot d_{h,G}(u,v)$.
\end{definition}



\paragraph{High-Level Algorithm.}
As the algorithm may be of interest beyond our application, we first describe it at a high level and then provide the details on how to efficiently implement the algorithm in our model.
Assume that we are given a graph $G=(V,E,w)$, a set of marked nodes $M\subseteq V$, and a hop distance parameter $h\geq 1$. 
Let us further assume that for all $e\in E$, we have $1\leq w(e)\leq W/h$ for some given $W\geq h$, so that the length of any path consisting of at most $h$ hops is between $1$ and $W$. 
The algorithm further has two parameters $k\geq 2$ and $\eta>1$ that control the stretch and the number of edges of the resulting spanner.

The algorithm consists of $\lceil\log_\eta W\rceil$ stages. In
the following, we focus on a specific stage $i\geq 1$. For
convenience, we define $L_i:=\eta^{i}$. The objective of stage $i$
is to construct a subset of the edges of $H=(M, E_H)$ that provides a good
approximation for any two nodes $u,v\in M$ for which the $h$-limited
distance in $G$ is in the range $[L_i/\eta,L_i]$. The final spanner is
then obtained by taking the union of the edges computed in the
individual stages.

Each stage consists of $k$ phases, which we number by $j=0,1,\dots,k-1$. Initially, all nodes in $M$ are active. We will see that nodes in $M$
become inactive as soon as it is guaranteed that all their $h$-limited
distances in the target range are already approximated well
enough. In the following, for a node $r\in G$, an integer parameter
$x\geq 1$, and a distance $L\geq 1$, we define
\[
B_G(r,x,L) := \set{v\in V\,:\,d_{h \cdot x,G}(r,v) \leq x\cdot L}
\]
The details of the algorithm for stage $i$ 
are given in \Cref{alg:skeletonSpanner}. We refer to the $k$ iterations of the
outermost for-loop as the $k$ phases $j=0,\dots,k-1$.

\begin{algorithm}[t]	
	\caption{\texttt{Skeleton Spanner Algorithm -- Stage $i$}}
	\label{alg:skeletonSpanner}
	\begin{algorithmic}
		\State \Comment{stage $i$ dealing with $h$-limited distances $\in[L_i/\eta,L_i]$}
		\State $V_0 := V$ \Comment{$V_j$ is the set of nodes that are active in phase $j$}
		\For{$j := 0$ \textbf{to} $k-1$}
		\State $G_j := G[V_j]$ \Comment{subgraph of $G$ induced by the active nodes} 
		\For{\textbf{each} $r\in V_j\cap M$} \Comment{for all active nodes in $M$}
		\State randomly sample $r$ with probability $|M|^{\frac{j+1}{k}-1}$ 
		\EndFor
		\For{\textbf{each} sampled $r\in M$}
		\For{\textbf{each} $v\in B_{G_j}(r,k-j,L_i)\cap M$}
		{add edge $(v,r)$ of weight $d_{h(k-j),G_j}(v,r)$ to $E_H$}
		\EndFor
		\State $V_{j+1}:= V_j \setminus B_{G_j}(r, k-j-1, L_i)$
		\EndFor
		\EndFor
	\end{algorithmic}			
\end{algorithm}

\begin{lemma}\label{lemma:spannerstretch}
	When a node $u\in M$ gets deactivated in stage $i$, for every $v\in M$ for which
	$d_{h,G}(u,v)\leq L_i$, the algorithm has added a path of length at most
	$2 k L_i$ to the spanner edge set $E_H$. Furthermore, this path consists of at most $2$ edges.
\end{lemma}
\begin{proof}
	Let $u,v\in M$ be two nodes for which $d_{h,G}(u,v)\leq L_i$ and let
	us show that the algorithm adds a path between $u$ and $v$ of length
	at most $2k L_i$ and consisting of at most $2$ edges---we call such a path a $(\leq 2)$-hop path in the following---to the spanner.  W.l.o.g., assume that
	$u$ is deactivated in phase $j$ and that $v$ is deactivated
	in a phase $j'\geq j$. If the algorithm has already added an $(\leq 2)$-hop path of
	length at most $2 k L_i$ between $u$ and $v$ prior to phase
	$j$, we are done.  Otherwise, we show that a) if there is not already
	a $(\leq 2)$-hop path of length at most $2 k L_i$ in the spanner connecting $u$
	and $v$, in the graph $G_j$ of the active nodes in phase $j$, the
	$h$-limited distance between $u$ and $v$ is still at most $L_i$ and
	b) in this case, in phase $j$, the algorithm adds a $(\leq 2)$-hop path of length at most
	$2 k L_i$ between $u$ and $v$ to the spanner.
	
	Let $P$ be a path of hop length $|P|\leq h$ and weight
	$w(P)\leq L_i$ connecting $u$ and $v$ in $G$. First assume that all
	nodes of $P$ are still active in phase $j$. We then clearly have
	$d_{h,G_j}(u,v)\leq L_i$. As node $u$ gets deactivated in phase $j$,
	we have that $u\in B_{G_j}(r,k-j-1,L_i)$ for some sampled node
	$r\in M$. We thus add an edge $(u,r)$ of length
	$d_{h(k-j-1),G_j}(r,u) \leq (k-j-1)L_i$ to $E_H$. Because
	$u\in B_{G_j}(r,k-j-1,L_i)$ and because $d_{h,G_j}(u,v)\leq L_i$, we
	can further conclude that $v\in B_{G_j}(r,k-j,L_i)$. We thus also
	add an edge $(v,r)$ of length
	$d_{h(k-j),G_j}(r,v) \leq (k-j) L_i$ to $E_H$. Together, the two
	edges thus provide a $(\leq 2)$-hop path of length at most
	$(2(k-j)-1) L_i< 2 k L_i$ between $u$ and $v$.
	
	It remains to consider the case that some nodes of $P$ are deactivated
	before phase $j$. Let $j'<j$ be the first phase, where some  node of
	$P$ is deactivated and let $w\in V$ be some node of $P$ that is
	deactivated in phase $j'$. This implies that there is some node
	$r'\in M$ such that $w\in B_{G_{j'}}(r,k-j',L_i)$. Because the path
	$P$ is completely contained in $G_{j'}$, we have
	$d_{h,G_{j'}}(w,u)\leq L_i$ and $d_{h,G_{j'}}(w,v)\leq L_i$. Both
	nodes $u$ and $v$ are thus contained in $B_{G_{j'}}(r,k-j',L_i)$ and
	thus in phase $j'$, the algorithm adds edges $(u,r')$ and $(v,r')$
	of length at most $(k-j')L_i\leq k L_i$ to $E_H$ and thus there is a
	$(\leq 2)$-hop path of length at most $2k L_i$ between $u$ and $v$ in the
	constructed spanner.
\end{proof}

The following lemma shows that in each phase, every node is only involved in the distance computations for few randomly centers. This on the one hand implies that the spanner algorithm does not add too many edges, and on the other hand it also allows to implement the algorithm efficiently by using only global edges. The lemma follows because of the radius of the balls that are contacted by each randomly selected center decreases from phase to phase such that the radius at which nodes are deactivated in phase $j$ is the same as the radius in which nodes are contacted in the phase $j+1$ and thus essentially, if a node $v$ expects to ``see'' many centers in phase $j+1$, the node should have been deactivated in phase $j$. A similar argument has previously been used by Blelloch et al.\ in \cite{blelloch14}.

\begin{lemma}\label{lemma:spannerprob}
	W.h.p., in every phase $j$ of \Cref{alg:skeletonSpanner},
	every node $v\in V_j$ is in $B_{G_j}(r,k-j,L)$ for at most $O(|M|^{1/k} \log
	n)$ sampled nodes $r\in M$.
\end{lemma}
\begin{proof}
	For phase $j\geq 0$ and a node $v\in V_j$, let
	$R_{v,j}:=\set{r\in M : v\in B_{G_j}(r,k-j,L_i)}$ and let
	$R_{v,j}':=\set{r\in M : v\in B_{G_j}(r,k-j-1,L_i)}$. The nodes in
	$R_{v,j}$ are the ones that, when sampled in phase $j$, reach node
	$v$, whereas $R_{v,j}'$ contains the set of nodes that, when samples
	in phase $j$, reach and also deactivate $v$. In phase $j$, nodes of
	$M$ are sampled with probability $p_j:=|M|^{\frac{j+1}{k}-1}$. We
	need to show that w.h.p.,
	$p_j\cdot |R_{v,j}|=O(|M|^{1/k}\cdot \log n)$ for all $v\in V$ and
	all phases $j$. The lemma then follows by a standard Chernoff bound
	and a union bound over all $v$ and $j$.
	
	To prove that $p_j\cdot |R_{v,j}|=O(|M|^{1/k}\cdot \log n)$, we show
	that otherwise, $v$ would have been deactivated in the previous
	phase, w.h.p. In the following, let $c>0$ be a constant that will be
	determined at the end. For a node $v\in V_j$ and a phase $j$, let
	$\calE_{v,j}$ be the event that $p_j\cdot |R_{v,j}'|\geq c\ln n$ and
	that node $v$ is \emph{not} deactivated in phase $j$. Recall that
	node $v$ is deactivated in phase $j$ if and only if one of the nodes
	in $R_{v,j}'$ is sampled in \Cref{alg:skeletonSpanner}. For all
	$v\in V_j$, we
	therefore have 
	\begin{equation}\label{eq:notdeactivated}
	\forall v\in V_j\,:\,\Pr(\calE_{v,j}) \leq (1-p_j)^{|R_{v,j}'|} <
	e^{-p_j|R_{v,j}'|} \leq \frac{1}{n^c}.
	\end{equation}
	Note that for $v\in V_j$ for which $p_j\cdot |R_{v,j}'|< c\ln n$, we
	have $\Pr(\calE_{v,j}) = 0$. Let us further define the random
	variable $X_{v,j}$ as the number of sampled nodes $r\in M$ in phase
	$j$ for which node
	$v\in V_j$ is in $B_{G_j}(r,k-j,L)$.  That is,
	$X_{v,j}$ counts the number of sampled nodes from
	$R_{v,j}$ in phase $j$. If we prove that $X_{v,j} = O(|M|^{1/k}\log
	n)$ w.h.p., the claim of the lemma follows by applying a union bound
	over all phases $j$ and all nodes $v\in V_j$.
	
	To study the number of sampled nodes from $R_{v,j}$, observe that
	$R_{v,j}\subseteq R_{v,j-1}'$. This follows from the definition of
	$R_{v,j}=\set{r\in M : v\in B_{G_j}(r,k-j,L_i)}$ and
	$R_{v,j-1}'=\set{r\in M : v\in B_{G_{j-1}}(r,k-j,L_i)}$ and the fact
	that $G_j$ is a subgraph of $G_{j-1}$. If we condition on the event
	$\overline{\calE}_{v,j-1}$, we know that $p_{j-1}|R_{v,j-1}'|< c\ln n$ as
	otherwise, $\overline{\calE}_{v,j-1}$ would have been deactivated in
	phase $j-1$ and thus $v\not\in V_j$. The sampling probabilities
	increase by a factor $|M|^{1/k}$ from phase to phase and thus,
	conditioning on $\overline{\calE}_{v,j-1}$ implies that
	$p_j|R_{v,j}|\leq p_j|R_{v,j-1}'| < c |M|^{1/k}\ln n$. We therefore have
	$\E[X_{v,j}] < c |M|^{1/k}\ln n$. A standard Chernoff argument thus
	shows that 
	\begin{equation}\label{eq:Xvj_bound}
	\Pr\left(X_{v,j} > 2c|M|^{1/k}\ln n\,\big|\,
	\overline{\calE}_{v,j-1}\right) \leq e^{-c/3 \cdot |M|^{1/k}\cdot
		\ln n} 
	\ \stackrel{(|M|^{1/k}>1)}{<}\ \frac{1}{n^{c/3}}.
	\end{equation}
	Let $\calB_{v,j}$ be the event that $X_{v,j}> 2c|M|^{1/k}\ln n$. By using the law of
	total probability, we then get
	\begin{eqnarray*}
		\Pr(\calB_{v,j}) 
		& = &
		\Pr\big(\calB_{v,j}\,\big|\,\overline{\calE}_{v,j-1}\big)\cdot \Pr(\overline{\calE}_{v,j-1})
		\ +\ 
		\Pr\big(\calB_{v,j}\,\big|\,{\calE}_{v,j-1}\big)\cdot
		\Pr({\calE}_{v,j-1})\\
		& \leq &
		\Pr\big(\calB_{v,j}\,\big|\,\overline{\calE}_{v,j-1}\big)
		+ \Pr({\calE}_{v,j-1})\\
		& \stackrel{\eqref{eq:notdeactivated},\eqref{eq:Xvj_bound}}{\leq}
		&
		\frac{1}{n^{c/3}} + \frac{1}{n^c},                                                                                       
	\end{eqnarray*}
	which concludes the proof.
\end{proof}

We now have everything we need in order to prove the main property of
the described spanner algorithm.

\begin{lemma}\label{lemma:spanneralg}
	Given a weighted graph $G=(V,E,w)$, a set of marked nodes
	$M\subseteq V$, as well as parameters $h\geq 1$, $k\geq 2$, and
	$\eta>1$,  the described spanner algorithm computes an $h$-hop skeleton
	spanner $H=(M, E_H)$ with stretch $2\eta k$. W.h.p., the number of
	edges $|E_H|$ is at most $|E_H|=O(k\cdot |M|^{1+1/k}\log n \cdot
	\log_\eta W)$. Further, for any two nodes $u,v\in M$ at hop distance at most $h$ in $G$, the spanner contains a path of hop length at most $2$ and of (weighted) length at most $2\eta k d_{h,G}(u,v)$.
\end{lemma}
\begin{proof}
	First note that by construction, as the weight of every edge that we
	add to $H$ corresponds to a path in $G$, we have $d_H(u,v)\geq
	d_{h,G}(u,v)$ for all $u,v\in M$. Also note that at the end of
	\Cref{alg:skeletonSpanner}, all nodes are inactive. This
	follows because for $j=k-1$, the sampling probability is set to $1$
	and therefore in the last phase, each remaining node in $M$ is a
	sampled. 
	
	To prove the stretch bound, it remains to show that
	$d_H(u,v)\leq 2\eta k\cdot d_{h,G}(u,v)$. We will at the same time also show that any two nodes within hop distance $h$ in $G$ will be connected by a $(\leq 2)$-hop path of this length. Let us first consider a
	single stage $i$. \Cref{lemma:spannerstretch} together with the
	fact that at the end, all nodes are inactive, implies that for
	any two nodes $u,v\in M$ for which $d_{h,G}(u,v)\leq L_i$, the
	spanner contains a $(\leq 2)$-hop path of length at most $2k L_i$. This provides a
	path of the right stretch for all node pairs $u,v\in M$ for which
	$d_{h,G}(u,v)\in [L_i/\eta,L_i]$. The stretch bound now directly
	follows because the spanner is defined as the union of the parts
	computed in each stage and because $\bigcup_{i=1}^{\lceil\log_\eta
		W \rceil}[L_i/\eta, L_i] \supseteq [1,W]$.
	
	To upper bound the number of edges of the spanner $H$, we again
	consider a single stage $i$. In each phase $j$ of stage $i$, each
	node $v\in V_j\cap M$ adds an edge to each sampled node $r\in M$ for which
	$v\in V_j\cap M$ is in $B_{G_j}(r,k-j,L)$. By Lemma
	\ref{lemma:spannerprob}, the number of such nodes is
	$O(|M|^{1/k}\log n)$, w.h.p. As the stage has $k$ phases, w.h.p., we therefore add
	at most $O(k |M|^{1/k}\log n)$ edges per node $v\in M$ and thus at
	most $O(k |M|^{1+1/k}\log n)$ edges in total. The lemma now follows
	because the total number of stages is $O(\log_\eta W)$.
\end{proof}

\paragraph{Realization in the Global Network.} 
We now describe how the algorithm can be efficiently implemented in our model.
As our algorithm exclusively relies on the global network, we can again use techniques from \cite{AGG+18}.
We assume that graph $G$ is given in \emph{$\delta$-oriented form}:
Every edge is only known by one its endpoints, which is \emph{responsible} for the edge, and every node in $G$ is responsible for at most $\delta$ edges.
We construct $H$ in a similar form:
When $v \in B_{G_j}(r, k-j, L_i) \cap M$ in some phase of the algorithm, the edge $(v,r)$ is added to $E_H$ only by $v$, which becomes responsible for the edge, and without the knowledge of $r$.

We achieve this by essentially performing limited-depth BFS constructions as in \cite{AGG+18}, and refer the reader to the paper for the technical details of our algorithm.
In phase $j$ of the algorithm, our goal is to let every node $r \in M$ inform all nodes in $B_{G_j}(r, k-j, L_i)$ about their $h(k-j)$-limited distance to $r$.
From a high level, we achieve this by propagating distance information from all nodes in $M$ for $h(k-j)$ iterations in $G_j$.
Whenever a message passes an edge of $G_j$, its distance value is increased by the weight of that edge, and is dropped, if that weight exceeds $(k-j) \cdot L_i$.
After $h(k-j)$ iterations, every node $v \in M$ knows $d_{x, G_j}(v,r)$ for all $x \le h(k-j)$ and all $r \in M$, if that value is at most $(k-j) \cdot L_i$.
In particular, $v$ can conclude if $v \in B_{G_j}(r, k-j, L_i)$, in which case $(v,r)$ must be added to $E_H$, and if $v \in B_{G_j}(r, k-j-1, L_i)$, in which case it becomes inactive.

As in \cite{AGG+18}, we realize one iteration of message passing by performing multi-aggregations in a butterfly, but have to take care of a few difficulties. 
First, in each phase $j$ we require broadcast trees that connect each node with its neighbors in $G_j$; therefore, each node $u$ learns which of its neighbors $v$ such that $u$ is responsible for $\{u,v\}$ is still active.
We achieve this by using multicast trees with multicast groups, which, for each node $v \in V$, contain all neighbors $u$ of $v$ such that $u$ is responsible for $\{u,v\}$.
We construct the trees already at the beginning of the algorithm, and use them prior to each phase to update nodes about neighbors that have become inactive in the previous phase.
Then, the nodes construct broadcast trees that connect all active neighbors with each other as in Lemma~16.
Here, we use the fact that every node is only responsible for at most $\delta$ edges for which it has to inject packets into the butterfly.
To allow for updating distance messages by edge weights, we additionally annotate each packet with the weight of the edge it corresponds to; thereby, each leaf node of every broadcast tree knows the weight of the corresponding edge, and can update the distance values of messages accordingly.
When the distance value of a message exceeds $(k-j) \cdot L_i$, it is simply dropped.

In iteration $t$ of the message passing process, every node $v$ participates in the BFS construction for each $r$ such that $d_{t-1, G_j}(v, r) \le (k-j) \cdot L_i$, and has to deliver a message to its neighbors for each construction.
Then, instead of injecting a single packet of size $O(\log n)$ into the butterfly, each node injects as many sub-packets as the number of constructions it participates in.
Sub-packets are forwarded sequentially throughout the multi-aggregation, such that each round of forwarding packets in the butterfly is "simulated" by performing multiple rounds of forwarding sub-packets.
As the number of sub-packets a packets consists of may vary, we synchronize each simulated round by performing convergecasts.
In the aggregation phase, two packets with the same destination may consist of sub-packets belonging to different constructions, which may result in a packet consisting of more sub-packets. For two sub-packets that belong to the same construction the one with smaller distance value is preferred.

\begin{lemma}\label{lemma:spannerreal}
Suppose $G=(V,E,w)$ is a weighted subgraph  of the global network given in $\delta$-oriented form and $M \subseteq V$. Then, an $h$-hop spanner as described in \Cref{lemma:spanneralg} in $O(k\cdot |M|^{1/k}\log n \cdot \log_\eta W)$-oriented form can be constructed in time $O((\delta + \log n)|M|^{1/k}\log^2 n \cdot h k^2 \cdot \log_\eta W)$, w.h.p.
\end{lemma}

\begin{proof}
	The correctness of the algorithm follows from the fact that in stage $i$ and phase $j$, every node $v \in V$ learns $d_{x, G_j}(v,r)$ for all $x \le h(k-j)$ and all $r \in M$, if that value is at most $(k-j) \cdot L_i$.
	From the discussion of \Cref{lemma:spanneralg} it directly follows that every node adds at most $O(k\cdot |M|^{1/k}\log n \cdot \log_\eta W)$ edges throughout the algorithm's execution, w.h.p., for which it becomes responsible.
	It remains to show the runtime of the algorithm, for which we refer to the results of \cite{AGG+18}.
	Setting up multicasts trees prior to the algorithm takes time $O(\delta + \log n)$ by Theorem 3.
	Now consider a single phase of the algorithm.
	Using the multicast trees to deliver updates about inactive neighbors takes time $O(\delta + \log n)$.
	Setting up the broadcast trees takes an additional $O(\delta + \log n)$.
	A single iteration of message passing in the BFS construction process takes $O(\delta + \log n)$ time of forwarding packets; however, this is slowed down by the number of sub-packets a packet may consist of, and an additional $\log n$ factor.
	By \Cref{lemma:spannerprob}, a node can only participate in at most $O(|M|^{1/k}\log n)$ BFS construction processes, and will therefore neither send out nor receive more than that many sub-packets in any multi-aggregation.
	Therefore, a single iteration takes time $O((\delta + \log n)|M|^{1/k}\log^2 n)$, w.h.p.
	In each phase, up to $hk$ iterations of message passing are performed.
	Multiplying this by the number of phases $k$ and the number of stages $\lceil \log_\eta W \rceil$ gives the stated bound.
\end{proof}

\subsubsection{The Recursive Algorithm} \label{sub:recurse}

Using the sparse skeleton spanner algorithm, we now present the algorithm to approximate SSSP.
The algorithm is divided into two stages.
The purpose of the first stage is to compute a hierarchical structure of spanners $G_1, \ldots, G_T$ as follows.
Let $G_0 := G$, and choose parameters $\alpha\geq 5$,  $h := c \alpha$ for a sufficiently large constant $c$, $k := \log_\alpha n$, and constant $\eta > 1$.
We construct the first sparse spanner $G_1$, which contains all nodes of $G$, by performing the distributed Baswana-Sen algorithm \cite{BS07} in the local network with parameter $k$, where any time a node adds an edge to the spanner, it becomes responsible for that edge.
By slightly modifying the analysis, it can be shown that thereby we obtain a $(2\log_\alpha n - 1)$-spanner in $O(\log_\alpha n \alpha \log n)$-oriented form, w.h.p., in time $O(\log^2_\alpha n)$.
Every other spanner $G_{i}$ ($i\geq 2$) is constructed as an $h$-hop skeleton spanner of $G_{i-1}$, where every node in $G_{i-1}$ joins the set $M_i$ of marked nodes with probability $\log(n)/\alpha$ for $i=2$ and with probability $1/\alpha$ for $i\geq 3$.
When for the first time a spanner $G_{T+1}$ contains no nodes anymore, the first stage of the algorithm ends.

After the first stage has finished, in the second stage we simply perform a BFS from $s$ in the union of all recursively constructed spanners $H = \bigcup_{1 \le j \le T} G_i$ for $O(\alpha \log n)$ rounds.
Finally, every node $v \in V$ chooses the minimum of all received distance values as its estimate $\tilde{d}(s,v)$ of $d(s,v)$.
In the following, we show that $H$ is a good spanner of the underlying graph $G$ and that, moreover, between any two nodes of $G$, there is a short path consisting of at most $O(\alpha\log n)$ hops in $H$, whose length gives a good distance approximation of the actual length of a shortest path in $G$.
We first need a technical lemma.

\begin{lemma}\label{lemma:shortpaths}
	Assume that $P$ is a shortest path on $G_1$ between two nodes of $G_1$. Further consider $i\geq 2$ and let $u,v \in M_i$ be two nodes on the path $P$ that are within $q$ hops for some $q\in[\gamma\cdot \alpha^{i-2},\gamma\cdot \alpha^{i-1}]$ for a sufficiently large constant $\gamma$. Then, $u$ and $v$ are connected in $G_{i-1}$ by a path $P$ such that $P$ consists of at most $O(\alpha)$ hops and it has length at most $(2\eta k)^{i-1}\cdot d_{G_1}(u,v)$.
\end{lemma}
\begin{proof}
	We prove the lemma by induction on $i$. For $i=2$, the statement holds directly as a consequence of \Cref{lemma:spanneralg} because in $G_2$, each node of $V$ (and thus of $G_1$) is sampled with probability $\log(n)/\alpha$ and because $G_2$ is an $h=c \alpha$-hop $(2\eta k)$-stretch skeleton spanner of $G_1$ w.r.t.\ the node set $M_2$.
	
	We can therefore focus on the induction step and $i\geq 3$. 
	Let $P'$ be the subpath of $P$ between $u$ and $v$. Further, let $M_{i-1}^{(P')}$ be the set of nodes of $M_{i-1}$ that are on path $P'$. Note that $u$ and $v$ are both in $M_{i-1}^{(P')}$ (because $M_i\subseteq M_{i-1}$). Note also that for all $i\geq 2$, nodes of $V$ are sampled to be in $M_i$ with probability $\log(n)/\alpha^{i-1}$. Because the hop-length $q$ of $P'$ is at least $\gamma\alpha^{i-2}$ for a sufficiently large constant $\gamma$, we have $M_{i-1}^{(P')}=\Omega(\log n)$, w.h.p. Further, because $q$ is upper bounded by $\gamma\alpha^{i-1}$, it also holds hat $M_{i-1}^{(P')}=O(\alpha\log n)$, w.h.p. Our goal is to select a subset $M_{i-1}'$ of $M_{i-1}^{(P')}$ of size $O(\alpha)$ such that $u,v\in M_{i-1}'$ and such that for any two consecutive nodes $x$ and $y$ of $M_{i-1}'$ on $P'$, it holds that the subpath of $P'$ connecting $x$ and $y$ is of length between $\gamma\alpha^{i-2}/5$ and $\gamma\alpha^{i-2}$. To see that this is always possible, we partition the path $P'$ into arbitrary subpaths that are all of length between $\gamma\alpha^{i-2}/5$ and $\gamma\alpha^{i-2}/4$. Because $P'$ is of length $\kappa\cdot \gamma\alpha^{i-2}/5$ for some $\kappa\geq 5$, we can always partition it into $\lfloor \kappa\rfloor$ subpaths of the required range. We then select a maximal set of non-adjacent subpaths, which contains the first and the last of the subpaths (the ones containing $u$ and $v$). Because $\gamma$ is a sufficiently large constant, each of the subpaths contains at least one node of $M_{i-1}^{(P')}$. We add $u$ and $v$ and an arbitrary node from each other selected subpath to $M_{i-1}$, which gives a set $M_{i-1}$ with the required properties.
	
	Because any two consecutive nodes $x$ and $y$ in $M_{i-1}'\subset M_{i-1}$ are at distance at most $\gamma\alpha^{i-2}$ and of size at least $\gamma\alpha^{i-2}/5\geq\gamma\alpha^{i-3}$, we can apply the induction hypothesis to the subpath between $x$ and $y$ and conclude that $x$ and $y$ are connected in $G_{i-2}$ by a path consisting of $O(\alpha)$ hops and of total length at most $(2\eta k)^{i-2}\cdot d_{G_1}(x,y)$. Assume that the constant $c$ in the definition of $h=c\alpha$ is chosen sufficiently large that the hop-length of this $G_2$-path is at most $h$. By \Cref{lemma:spanneralg}, $x$ and $y$ are therefore connected by a path consisting of at most two hops and of length $(2\eta k)d_{h,G_{i-2}}(x,y)$. Note that the induction hypothesis then also implies that $d_{h,G_{i-2}}(x,y)\leq (2\eta k)^{i-2}d_{G_1}(x,y)$. and the claim of the lemma thus follows.
\end{proof}

We next prove that in the union spanner graph $H$, there is a short (in terms of hops and weight) path between any two nodes.

\begin{lemma}\label{lemma:hierarchicalstretch}
	Let $u,v\in V$ be two nodes of $G_1$ and let $P$ be a shortest path between $u$ and $v$ on $G_1$. 
	Assume that $P$ consists of $q$ hops and assume that for the constant $\gamma$  from \Cref{lemma:shortpaths}, $\xi$ is the smallest integer for which $q\leq \gamma\alpha^\xi$. 
	Then graph $H$ contains a path that consists of at most $O(\alpha\log_\alpha n)$ hops and that has total weight at most $(2\eta k)^{\xi-1} d_{G_1}(u,v)$.
\end{lemma}
\begin{proof}
	We prove the lemma by induction on $\xi$.
	First note that for $\xi=1$, the claim directly follows because $P$ is a path of length $O(\alpha)$ on $G_1$. Let us therefore consider $\xi\geq 2$. Recall that the nodes in $M_{\xi}$ are sampled with probability $\log(n)/\alpha^{xi-1}$. Hence, w.h.p., for sufficiently large $\gamma$, every subpath of length at most $\gamma\alpha^{\xi-1}/3$ contains at least one node of $M_{\xi}$. Let $x$ and $y$ be the first and the last node of $M_{\xi}$ on $P$ when going along the path from $u$ to $v$. Note that by the above observation, $x$ is within hop-distance $\leq \gamma\alpha^{\xi-1}/3$ from $u$ and $y$ is within hop-distance $\leq \gamma\alpha^{\xi-1}/3$ from $v$. Because the hop-length of $P$ is at most $\gamma\alpha^{\xi}$, clearly also the hop distance of the subpath $P'$ of $P$ between $x$ and $y$ at most $\gamma\alpha^{\xi}$. Let us first assume that the hop-length of $P'$ is at least $\gamma\alpha^{\xi-1}$. Then \Cref{lemma:shortpaths} implies that $G_{\xi-1}$ (and thus $H$) contains a path consisting of $O(\alpha)$ nodes and of total length at most $(2\eta k)^{\xi-1}\cdot d_{G_1}(x,y)$. Otherwise, we have that the number of hops of $P'\geq \gamma\alpha^{\xi-1}/3\geq \gamma\alpha^{\xi-2}$. We can then assume that $\xi\geq 3$ as otherwise, $P'$ is a path of length $O(\alpha)$ on $G_1$. Because $M_{\xi}\subseteq M_{\xi-1}$, \Cref{lemma:shortpaths} then implies that $G_{\xi-2}$ (and thus $H$) contains a path consisting of $O(\alpha)$ nodes and of total length at most $(2\eta k)^{\xi-2}\cdot d_{G_1}(x,y)$. In both cases, we have reduced the problem to a case that is covered by the induction hypothesis. The hop length of the combined path follows because we get at most $2$ path of length $O(\alpha)$ for each $\xi$-value.
\end{proof}

\begin{theorem} \label{theorem:polySSSP}
	The algorithm solves $(\log_\alpha n)^{O(\log_\alpha n)}$-SSSP in time $\tilO(\alpha^3)$, w.h.p.
\end{theorem}

\begin{proof}
	We first show the runtime of the first stage.
	Constructing the first spanner $G_1$ takes time $O(\log^2_\alpha n)$, which follows from \cite{BS07} and  our choice of $k = \log_\alpha n$.
	The resulting spanner is in $\tilO(\alpha)$-oriented form.
	By our choice of $k$, every spanner $G_{i}$, $i \ge 2$, that is constructed using our sparse skeleton spanner algorithm is also in $\tilO(\alpha)$-oriented form by \Cref{lemma:spannerreal}.
	Therefore, constructing $G_{i}$ takes time $\tilO(\alpha |M|^{1/\log_\alpha n} \alpha) = \tilO(\alpha^3)$.
	Furthermore, we have that $T = O(\log_\alpha n)$, which follows from the fact that the probability for a node to be in spanner $G_{i + 1}$ , $i+1 = c \log_\alpha n$ for some $c \ge 1$, is at most $\log n (1/\alpha)^{c \log_\alpha n} = 1/n^{c-1}$, and by applying the union bound over all nodes.
	Thus, the first stage of the algorithm takes time $\tilO(\alpha^3)$, w.h.p.
	
	Now consider the second stage. 
	As every spanner is in $\tilO(\alpha)$-oriented form, in $H$ every node is responsible for $\tilO(\alpha)$ edges.
	Therefore, the broadcast trees necessary to perform multi-aggregations as in \cite{AGG+18} can be set up in time $\tilO(\alpha)$, w.h.p., and every round of the BFS can be realized by performing multi-aggregations in time $\tilO(\alpha)$, w.h.p. 
	As we perform $\tilO(\alpha)$ rounds of BFS, the second stage takes time $\tilO(\alpha^2)$, w.h.p.

	Finally, we show the approximation factor.
	Let $u, v \in V$ be two nodes.
	By \Cref{lemma:hierarchicalstretch} and because the total number of recursive levels is at most $T=O(\log_\alpha n)$, between any two nodes $u,v\in V$, the combined spanner $H$ contains a path of length at most $(2\eta k)^T\cdot d_{G_1}(u,v)=O(k)^{O(\log_\alpha n)}\cdot d_{G_1}(u,v)$ and this path consists of at most $O(\alpha\log_\alpha n)$ hops. 
	Because also the spanner $G_1$, which is computed by using the algorithm of Baswana and Sen~\cite{BS07} has hop and distance stretch $O(k)$, we can conclude that between any two nodes $u,v\in V$, the spanner $H$ contains a path of length $(\log_{\alpha}n)^{O(\log_\alpha n)}\cdot d_G(u,v)$, consisting of at most $O(\alpha\log_\alpha n) = \tilO(\alpha)$ hops.
	This particularly shows that by propagating distances from $s$ for $O(\alpha \log_\alpha n)$ rounds, every node $v \in V$ learns a distance estimate $\tilde{d}(s,v) \le (\log_{\alpha}n)^{O(\log_\alpha n)}\cdot d(s,v)$, which concludes the proof.
\end{proof}


 \clearpage
 
 \appendix 

\section{General Notions from Probability Theory}
\label{apx:generalnotations}


\begin{lemma}[Chernoff Bound]
	\label{lem:chernoffbound}
	We use the following forms of Chernoff bounds in our proofs: 
		$$\mathbb{P}\big(X > (1 \!+\! \delta) \mu_H\big) \leq \exp\big(\!\!-\!\frac{\delta\mu_H}{3}\big),$$ 
	with $X = \sum_{i=1}^n X_i$ for i.i.d.\ random variables $X_i \in \{0,1\}$ and $\mathbb{E}(X) \leq \mu_H$ and $\delta \geq 1$. Similarly, for $\mathbb{E}(X) \geq \mu_L$ and $0 \leq \delta \leq 1$ we have
	$$\mathbb{P}\big(X < (1 \!-\! \delta) \mu_L\big) \leq \exp\big(\!\!-\!\frac{\delta^2\mu_L}{2}\big).$$	
\end{lemma}

\begin{lemma}[Union Bound]
	\label{lem:unionbound}
	Let $E_1, \ldots ,E_k$ be events, each taking place w.h.p. If $k \leq p(n)$ for a polynomial $p$ then $E \coloneqq \bigcap_{i=1}^{k} E_i$ also takes place w.h.p.
\end{lemma}

\begin{proof}
	Let $d \coloneqq \deg(p)\!+\!1$. Then there is an $n_0 \geq 0$ such that $p(n) \leq n^d$ for all $n \geq n_0$. Let $n_1, \ldots , n_k \in \mathbb{N}$ such that for all $i \in \{1, \ldots, k\}$ we have $\mathbb{P}(\overline{E_i}) \leq \tfrac{1}{n^c}$ for some (yet unspecified) $c > 0$.	
	With Boole's Inequality (union bound) we obtain
	\begin{align*}
	\mathbb{P}\big(\overline{E}\big) = \mathbb{P}\Big(\bigcup_{i=1}^{k} \overline{E_i} \Big) \leq \sum_{i=1}^{k} \mathbb{P}(\overline{E_i}) \leq \sum_{i=1}^{k} \frac{1}{n^c} \leq \frac{p(n)}{n^{c}} \leq \frac{1}{n^{c-d}}				
	\end{align*}
	for all $n \geq n_0' \coloneqq \max(n_0, \ldots ,n_k)$. Let $c' > 0$ be arbitrary. We choose $c \geq c' \!+ d$. Then we have $\mathbb{P}\big(\overline{E}\big) \leq \frac{1}{n^{c'}}$ for all $n \geq n_0'$.
\end{proof}

\begin{remark}
	If a finite number of events is involved we use the above lemma without explicitly mentioning it. It is possible to use the lemma in a nested fashion as long as the number of applications of the lemma is polynomial in $n$.
\end{remark}

%

\section{A Scheme to Balance Congestion}
\label{apx:balance_congestion}

The following lemma is a slight adaptation of \cite{Gha15}, Theorem 1.1.
We say an algorithm $\mathcal{A}$ is \emph{simple}, if during its execution only local edges are used, at most $O(1)$ messages are sent over every edge in each round, and its execution only depends on $G$ and the node's input for $\mathcal{A}$, i.e., $\mathcal{A}$ is independent from any other concurrently running algorithm.

\begin{lemma}
	\label{lem:balance_congestion}
	Let $\mathcal{A}_1, \ldots, \mathcal{A}_k$ be simple algorithms and let $D$ be the maximal running time of any algorithm.
	Further, let $C$ be the maximum cumulative number of messages that are sent over some edge by executing $\mathcal{A}_1, \ldots, \mathcal{A}_k$. If local capacity $\lambda \ge \log n$, then there is an algorithm that executes $\mathcal{A}_1, \ldots, \mathcal{A}_k$ in time $O(C/\lambda + D + \log n)$, w.h.p.
\end{lemma}

\begin{proof}
	The idea of the algorithm is to begin the execution of $\mathcal{A}_i$ in round $t_i$, which is chosen uniformly at random from the interval $[1, \alpha C/\lambda]$ for some constant $\alpha \ge 3$.
	As pointed out in \cite{Gha15}, sharing $O(\log^2 n)$ bits of randomness suffices to obtain $\Theta(\log n)$-wise independence for the choice of $t_i$'s.
	These bits can easily be broadcasted in the global network in time $O(\log n)$ (see, e.g., \cite{AGG+18}).
	Clearly, as $t_i \le \alpha C/\lambda$ for all $i$, the algorithm takes time $O(C/\lambda + D + \log n)$.
	It remains to show that in every round of the algorithm only $\lambda$ messages need to be sent over the same local edge.
	
	Fix some edge $e \in E$ and round $t$, and let $X_{i}$ be the binary random variable that is $1$ if and only if messages are sent over $e$ in round $t$ in algorithm $\mathcal{A}_i$.
	Further, let $c_i$ be the number of rounds in which messages are sent over $e$ in $\mathcal{A}_i$.
	By the union bound, we have that $Pr[X_{i} = 1] \le c_i  \lambda / (\alpha C)$.
	Then $X = \sum_{i = 1}^k X_i$ is a sum of $\Theta(\log n)$-wise independent binary random variables with expected value $E[X] \le \sum_{i=1}^k c_i \lambda / (\alpha C) = \lambda / \alpha$.
	Therefore, Theorem 5 (II) (a) of \cite{SSS95} with $\delta = 2$ and $k = \Theta(\log n) \le \lfloor 2 \lambda e^{-1/3} / \alpha\rfloor$ for sufficiently large $\alpha$ yields
	\[
	Pr[X \ge \lambda] \le Pr[X \ge 2 \lambda/\alpha] \le e^{-\Theta(\log n)}.
	\]
	Creating sufficiently large independence and taking the union bound over all edges and rounds implies the lemma.
\end{proof}

We can apply the above lemma to bound the local capacity required to learn the graph $G$ up to $h$ hops.

\begin{lemma}
	\label{lem:congestion_to_learn_G}
	Local capacity $\lambda = \Theta(|E|/h)$ suffices so that all nodes of $G$ can learn $G$ (including edge weights) up to hop-distance $h$ in $\tilO(h)$ rounds via the local network.
\end{lemma}

\begin{proof}
	For each $e \in E$ we define an algorithm $\mathcal A_e$. Algorithm $\mathcal A_e$ runs on every node and does the following. In the first round every node $v$ that is adjacent to $e$ sends the information about $e$ (including the weight label) to all its neighbors via local edges. If node $v$ learns $e$ for the \textit{first time} in some round, then it sends $e$ to all its neighbors via local edges. Note that $e$ and its weight label fit into a message of size $O(\log n)$ (recall that weights are polynomially bounded in $n$). 
	
	Since each algorithm $\mathcal A_e$ sends at most two messages over each edge (once from each endpoint) and the algorithms $\mathcal{A}_e, \: e \in E $ are independent, they are simple. Clearly, by running all $\mathcal A_e$ in parallel for $\tilO(h)$ rounds, every node learns $G$ up to hop-distance $h$, so the dilation is $D = \tilO(h)$. Additionally, we send at most $C = 2|E|$ messages over any edge in any round.	
	According to the previous \Cref{lem:balance_congestion} it is possible execute all $\mathcal A_e, \: e \in E$ in $O(C/\lambda + D + \log n) = \tilO(|E|/\lambda + h)$ rounds. Thus  $\lambda = \Theta(|E|/h)$ suffices to achieve the running time $\tilO(h)$.	
\end{proof}

If nodes only want to learn the $h$-limited distances to a subset of $k$ nodes in $h$ rounds we can employ the distributed version of Bellman-Ford for $k$ sources (instead of learning the whole $h$-neighborhood) which sends at most $2k$ messages over each edge per round. Specifically this means, that all nodes can learn the $h$-limited distances to every node in their $h$-hop neighborhood in $h$ rounds, with local capacity $\lambda = 2n$.

\begin{lemma}
	\label{lem:congestion_bellman_ford}
	Local capacity $\lambda = 2k$ suffices so that all nodes of $G$ can learn their $h$-limited distances to a subset of $k$ nodes in $h$ rounds via the local network.\footnote{Note that this result can be improved even further by adapting the ``short-range'' \CONGEST algorithm of \cite{HNS17} to compute $h$-limited shortest paths for $k$ sources in $\tilO(k\sqrt{h})$ rounds.}
\end{lemma}

\begin{proof}
	The Bellman-Ford algorithm is rather simple, we sketch the distributed version in the following. Each node has a list of size $k$ containing the shortest distances to sources $1, \ldots ,k$ it knows so far (initially the list contains $\infty$ for all sources except those who are direct neighbors). Then each round, every node sends each of its neighbors its current list. With the respective lists a node receives from its neighbors, each node can update its own distance list. In round $i$ of the algorithm, every node $u$ knows $d_i(u,s)$ for all sources $s$. We send at most two lists over each edge (one from each endpoint) each round, so $\lambda = 2k$ suffices.
\end{proof}

 \clearpage

\bibliographystyle{abbrv}
\bibliography{references}

\end{document}